\documentclass{ws-ijmpb}
\usepackage{amsfonts,amssymb,amsmath}
\usepackage{graphicx}
\usepackage[super,sort,compress]{cite}
\usepackage{hyperref}
\usepackage{mathrsfs}

\usepackage{chngcntr}

\counterwithout{theorem}{section}
\counterwithout{example}{section}
\counterwithout{definition}{section}

\begin{document}

\markboth{Danko D. Georgiev}{Causal potency of consciousness in the physical world}

\catchline{}{}{}{}{}

\title{Causal potency of consciousness in the physical world}

\author{Danko D. Georgiev}

\address{Institute for Advanced Study, 30 Vasilaki Papadopulu Str.\\
Varna, 9010, Bulgaria\\
danko@q-bits.org}

\maketitle

\begin{history}
\received{(11 March 2023)}
\revised{(24 April 2023)}
\accepted{(5 May 2023)}
\end{history}

\begin{abstract}
The evolution of the human mind through natural selection mandates that our conscious experiences are causally potent in order to leave a tangible impact upon the surrounding physical world. Any attempt to construct a functional theory of the conscious mind within the framework of classical physics, however, inevitably leads to causally impotent conscious experiences in direct contradiction to evolution theory. Here, we derive several rigorous theorems that identify the origin of the latter impasse in the mathematical properties of ordinary differential equations employed in combination with the alleged functional production of the mind by the brain. Then, we demonstrate that a mind--brain theory consistent with causally potent conscious experiences is provided by modern quantum physics, in which the unobservable conscious mind is reductively identified with the quantum state of the brain and the observable brain is constructed by the physical measurement of quantum brain observables. The resulting quantum stochastic dynamics obtained from sequential quantum measurements of the brain is governed by stochastic differential equations, which permit genuine free will exercised through sequential conscious choices of future courses of action. Thus, quantum reductionism provides a solid theoretical foundation for the causal potency of consciousness, free will and cultural transmission.
\end{abstract}


~


\newpage
\tableofcontents

\newpage
\section{Introduction}

\subsection{The sense of agency}

We are \emph{sentient beings} that possess an \emph{inner psychological world} or a \emph{stream of conscious experiences}, which we simply refer to as a \emph{mind} \cite{James1890,McGilvary1907,Nagel1974}.
In fact, we are what our conscious minds are.
It is only through our conscious experiences that we are able to access the surrounding world, comprehend it and act upon it \cite{Georgiev2017}.
Our personal thoughts, aims, goals and desires motivate us to strive towards achieving a healthy, prosperous and happy life \cite{Peterson2005}.
The subjective awareness of initiating, executing, and controlling our volitional actions in the physical world corroborates daily our \emph{sense of agency}.
In the absence of conscious experiences, however, such as during general anesthesia \cite{Rudolph2004,Franks2008,Brown2011},
syncope \cite{Kapoor2000} or coma \cite{Brown2010}, we lose our
\emph{feeling of agency} within the world. Furthermore, exactly because we
are conscious agents with causative potency, it is possible for civil
law to establish blameworthiness for actions that are considered wrongdoings
\cite{Fagelson1979,Jaffey1985,Hopkins2005} and ethics can hold us
morally responsible for what ensues from our behavior \cite{Sartre2007,Shafer2020}.

\subsection{The problem of mental causation}

The necessity to describe our conscious minds within the framework
of available physical theories, however, inevitably confronts us with
the problem of \emph{mental causation}, namely, how is it possible
and what is the physical mechanism through which our consciousness
is able to affect the physical world. Apparently, branding the conscious
mind as ``non-physical'' is not going to be helpful because
it will immediately put the mind outside of physics, and consequently
if the mind is not subject to any physical laws it would be impossible
to derive any mental causation upon the physical world.

Because physics is the most fundamental scientific discipline, it
is expected to encompass \emph{everything in existence} and study
the \emph{entirety of reality} using mathematical principles \cite{Feynman2013}.
Our conscious minds do exist, therefore they are real and have to be defined as
``physical'' \cite{Georgiev2020a,Georgiev2020b} thereby enjoying
the privilege to be considered a valid subject for discussion by physical
theories. Moreover, if a physical theory does not include consciousness
or makes incorrect predictions with regard to consciousness, then
such a physical theory should be deemed either \emph{falsified} or \emph{incomplete}.
Therefore, conscious experiences should be considered \emph{physical} and their causal effects upon the physical world should be described by a \emph{physical theory}---either a physical theory that we already have or a physical theory that will be constructed in the future.
The focus of the present work will be on comparison of available \emph{classical} or \emph{quantum} physical theories in view of demonstrating that quantum physics already has all the necessary mathematical ingredients for accommodating a causally potent consciousness.

\subsection{Evolutionary theory mandates causally potent consciousness}

The evolution of consciousness and development
of culture in primates \cite{Imanishi1951,Imanishi1952,Imanishi1957,Imanishi1966,Matsuzawa2001,Gruber2012},
including chimpanzees \cite{Goodall1986,Whiten1999,Gruber2009,Lind2010,Boesch2020},
bonobos \cite{Samuni2020,vanLeeuwen2020}, gorillas \cite{Robbins2016},
orangutans \cite{Krutzen2011,Gruber2012b,Koops2014,Lameira2022},
snow monkeys \cite{Kawamura1959,Kawai1965,Leca2012} and early humans
\cite{Gross2020,Ehrenreich2019,Wade2006,Melcher2006,Harris2011,Herzog2010,Conard2009a,Conard2009b,Valladas2001},
is an established empirical fact that entails the causal potency of
consciousness in the physical world \cite{Popper1983,Eccles1991,Eccles1992,Eccles1994}.
From this stronghold, we will investigate the implications of
different physical approaches to modeling consciousness, including
\emph{functionalism} or \emph{reductionism}, within the context of
\emph{classical} or \emph{quantum} physics. Utilizing the
mathematical properties of ordinary differential equations or stochastic
differential equations, we will derive rigorous theorems, which constrain
the available solutions to the problem of mental causation.
In particular, we will show that within available physical theories, it is quantum physics that provides the only plausible framework
for a physical theory of causally potent consciousness and free will.

\subsection{A synopsis of the presentation on the causal potency of consciousness}

The subsequent presentation is structured as follows:
In Section~\ref{sec:2},
we formulate the problem of mental causation and briefly review the
neuroanatomy and neurophysiology of the brain cortex, peripheral neural
system, sensory organs and effector organs. Next,
in Section~\ref{sec:3}
we explain why the problem of mental causation is unsolvable in classical
physics. We also illustrate the mathematical theory of ordinary differential
equations and prove theorems that eliminate classical functionalism
and classical reductionism as plausible theories of consciousness.
Then,
in Section~\ref{sec:4} we illustrate the mathematical theory
of stochastic differential equations and prove theorems that recognize
quantum reductionism as the most plausible theory of consciousness.
In order to make the exposition self-contained,
in Section~\ref{S5} we elaborate on free will and its representation using stochastic processes,
in Section~\ref{S6} we illustrate how knowledge acquisition or learning generates dynamic biases and varying amounts of free will,
in Section~\ref{S7} we explain how quantum entanglement leads to mind binding and constrains free will,
and in Section~\ref{S8} we describe the physical mechanism underpinning the wave function collapse and disentanglement.
Lastly, in Section \ref{sec:9} we discuss the significance of the presented results and how they provide a consistent account of the natural evolution of the human mind.

\section{The causal potency problem}
\label{sec:2}

We react to sensory stimuli that are present in the surrounding world.
The neurophysiological account of our reaction starts with transduction
of the sensory stimulus into an electric signal by the sensory organs.
For example, the eyes convert visible light into electric currents
in the retina, the ears convert audible sound into electric currents
in the cochlea, and the skin converts mechanical pressure from touch
into electric currents in encapsulated nerve endings of Meissner corpuscles
\cite{Hall2021,Kandel2021}. The sensory information carried by these
electric signals then propagates along sensory pathways that reach
corresponding sensory areas in the brain cortex (Fig.~\ref{fig:1}A).
The voluntary decision for action is executed by the motor area of
the brain cortex, which sends motor electric signals toward the effector
organs such as the skeletal muscles that move the body (Fig.~\ref{fig:1}B).
Therefore, the overall process of reacting to sensory stimuli is a
sequential composition of inputting the sensory information from the
body to the brain cortex and outputting the voluntary choice of action
from the brain cortex to the body \cite{Georgiev2017}.

\begin{figure}[t!]
\begin{centering}
\includegraphics[width=\textwidth]{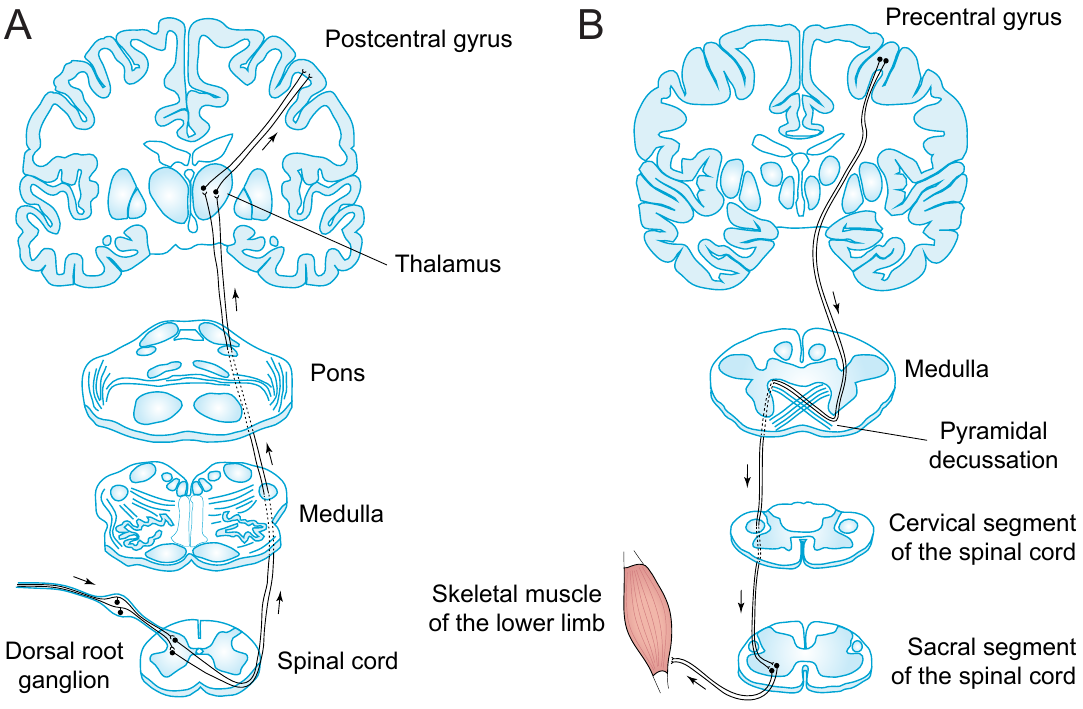}
\par\end{centering}
\caption{\label{fig:1}The human brain interacts with its physical environment. Electric signals in the form of action potentials mediate
the physiological communication between the human brain cortex and the body: the somatosensory pathway (A) delivers sensory information
from the body to the somatosensory cortex in the postcentral gyrus, whereas the somatomotor pathway (B) delivers motor information from
the motor cortex in the precentral gyrus to the body muscles. The spinal cord segments, medulla and pons are represented with their
transversal sections, whereas thalamus and cortex are shown in frontal slice. Modified from Ref.~\citen{Georgiev2017}.}
\end{figure}

Clinical observations from injuries of the human nervous system have
accumulated an overwhelming amount of evidence that the brain cortex
is the \emph{seat of consciousness} \cite{Ropper2019}. This is consistent
with physiological delays of at least 50 ms temporal duration between
the application of the sensory stimulus at the sensory organ and the
conscious perception of the sensation in the brain cortex or between
the conscious decision to elicit a voluntary movement in the brain
cortex and the actual contraction of the skeletal muscle \cite{Cattell-1,Cattell-2,Cattell-3}.

The somatosensory pathways from the body towards the brain cortex
and the somatomotor pathways from the brain cortex towards the body
could be either bypassed or replaced with brain--machine interfaces,
which deliver sensory information directly to the brain cortex where
they are consciously experienced \cite{Dobelle2000} or receive motor
information directly from the brain cortex for the control of robotic
devices \cite{Hochberg2006}. The engineering of brain--machine interfaces
and their practical application in neurorehabilitation is assisted by the layered organization
of pyramidal neurons inside the gray matter of the brain cortex \cite{Georgiev2020c}.

\begin{figure}[t!]
\begin{centering}
\includegraphics[width=\textwidth]{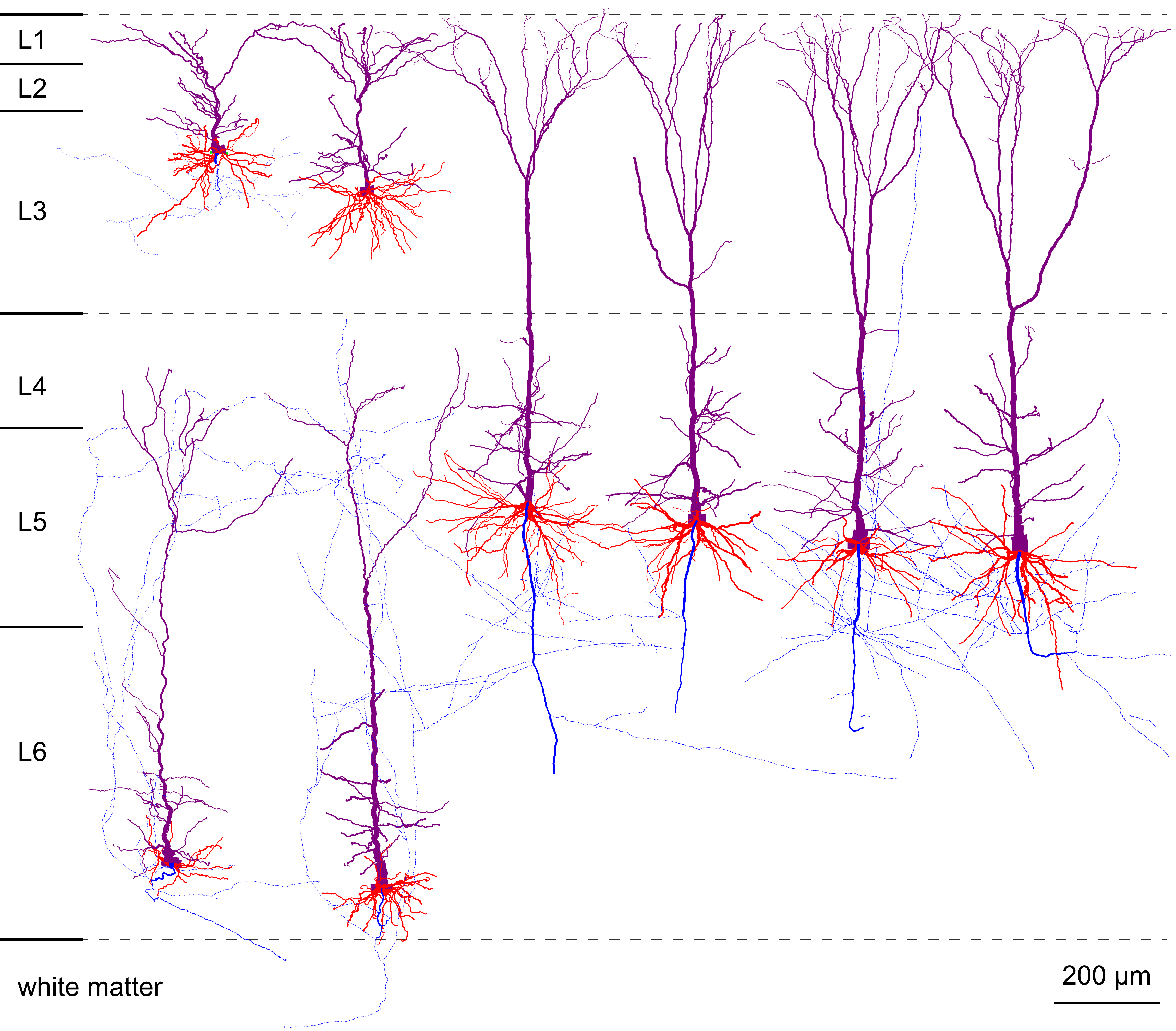}
\par\end{centering}

\caption{\label{fig:2}Layered structure (L1--L6) of the gray matter of rat
neocortex based on digital reconstructions from NeuroMorpho.org of
Layer~2-3 pyramidal neurons (\href{https://neuromorpho.org/neuron_info.jsp?neuron_id=NMO_49059}{NMO\_49059}, \href{https://neuromorpho.org/neuron_info.jsp?neuron_id=NMO_49054}{NMO\_49054}),
Layer~5 pyramidal neurons (\href{https://neuromorpho.org/neuron_info.jsp?neuron_id=NMO_77908}{NMO\_77908}, \href{https://neuromorpho.org/neuron_info.jsp?neuron_id=NMO_77904}{NMO\_77904}, \href{https://neuromorpho.org/neuron_info.jsp?neuron_id=NMO_77905}{NMO\_77905}, \href{https://neuromorpho.org/neuron_info.jsp?neuron_id=NMO_77920}{NMO\_77920}) and
Layer~6 pyramidal neurons (\href{https://neuromorpho.org/neuron_info.jsp?neuron_id=NMO_09382}{NMO\_09382}, \href{https://neuromorpho.org/neuron_info.jsp?neuron_id=NMO_64646}{NMO\_64646}).
Basal dendrites are rendered in red, apical dendrites in purple, and axons in blue.
Neuron identification numbers are listed from left to right in the order
the rendered reconstructions are assembled. Neuron reconstructions can be retrieved by their identification numbers at \url{https://neuromorpho.org/KeywordSearch.jsp}.}
\end{figure}

The microscopic neuroanatomy of a cross section of the brain cortex
reveals columns of vertically stacked pyramidal neurons \cite{Nieuwenhuys1994,Hoffmann2015,Anastassiou2015,Marx2012,Marx2015,Bekkers2011,Ascoli2006}
organized into 6 layers (Fig.~\ref{fig:2}). The first layer~L1 is
closest to the cortical surface, whereas the sixth layer~L6 lays deepest.
The dendrites of pyramidal neurons, which receive electric information,
extend towards the superficial cortical layers and form dense arborizations \cite{Karimi2020}.
The axons of pyramidal neurons, which output electric information,
extend towards the underlying white matter whose characteristic color
is due to the myelin sheets that insulate the axons from leaking their
ionic currents into nearby inactive axons \cite{Stadelmann2019}.

The collective excitation of columns of pyramidal neurons creates
multiple individual electric potentials that summate spatiotemporally into
a larger \emph{local field potential} that can be recorded with micro-electrodes
\cite{Long2021} implanted in the extracellular space of cortical
tissue \cite{Hwang2009,Hochberg2006}. If the recorded electric activity
from the brain cortex is forwarded to computer program that controls
a robotic arm, the conscious mind is able to train itself after several
months of practice to control the robotic arm without actually moving
any of the body muscles. Thus, our conscious minds appear to be causally
potent agents within the physical world, because if they were not,
conscious control of brain--machine interfaces would not have been
possible \cite{Georgiev2017,Georgiev2021c}. The problem of mental causation is
to explain how the conscious mind is able to physically affect the
electric activity of cortical pyramidal neurons. Whether such an explanation
is possible depends critically on the physical approach chosen for
addressing the mind--brain problem. We will elaborate on this in great detail next.

\section{Causal potency of consciousness in classical physics}
\label{sec:3}

The world of classical physics is based on two fundamental postulates.
First, it is assumed that all physical quantities are \emph{observable}.
This means that the physical states of classical systems can be measured
with physical instruments. Second, it is assumed that the \emph{time
dynamics} of physical states is governed by a \emph{system of ordinary
differential equations} (ODEs). This means that given an initial state
$S(0)$ of a classical system, its future time evolution $S(t)$ is
deterministic and can be calculated with absolute certainty and arbitrarily
high precision in principle.
\begin{example}
(Deterministic dynamics) The mathematical properties
of classical deterministic dynamics can be illustrated with the following
3-dimensional chaotic jerky Lorenz-like system whose time evolution is governed
by a single third-order ordinary differential equation (ODE) of a
time-dependent variable $x(t)$ \cite{Sprott2009}
\begin{equation}
\left(\frac{d}{dt}\right)^{3}x=-\left(\frac{d}{dt}\right)^{2}x-4\frac{d}{dt}x+5x-x^{3}
\label{eq:1}
\end{equation}
With the use of the following substitutions
\begin{eqnarray}
x_{1} & =& x \label{eq:2}\\
x_{2} & =& \frac{d}{dt}x=\frac{d}{dt}x_{1} \label{eq:3}\\
x_{3} & =& \left(\frac{d}{dt}\right)^{2}x=\frac{d}{dt}x_{2} \label{eq:4}\\
\frac{d}{dt}x_{3} & =& \left(\frac{d}{dt}\right)^{3}x
\end{eqnarray}
the single third-order ODE \eqref{eq:1} can be re-written as a system
of three first-order ODEs 
\begin{equation}
\begin{cases}
\frac{d}{dt}x_{1} & = x_{2}\\
\frac{d}{dt}x_{2} & = x_{3}\\
\frac{d}{dt}x_{3} & = -x_{3}-4x_{2}+5x_{1}-x_{1}^{3}
\end{cases}
\end{equation}
The main point emphasized here is that any single higher-order ODE
can always be re-written in mathematically equivalent form as a system
of several first-order ODEs. Once we have the system of first-order
ODEs, we can simulate the trajectory in time of the physical system
for any initial state $S(0)$ given as a set of observables $\{x_{1}(0),x_{2}(0),x_{3}(0)\}$
at the initial time $t=0$. The determinism of the computed dynamics
is manifested in the fact that no matter how many simulation runs
we perform (with the same initial state and a fixed level of precision), the resulting trajectory always remains the same (Fig.~\ref{fig:3}).
\end{example}
\begin{figure}
\begin{centering}
\includegraphics[width=\textwidth]{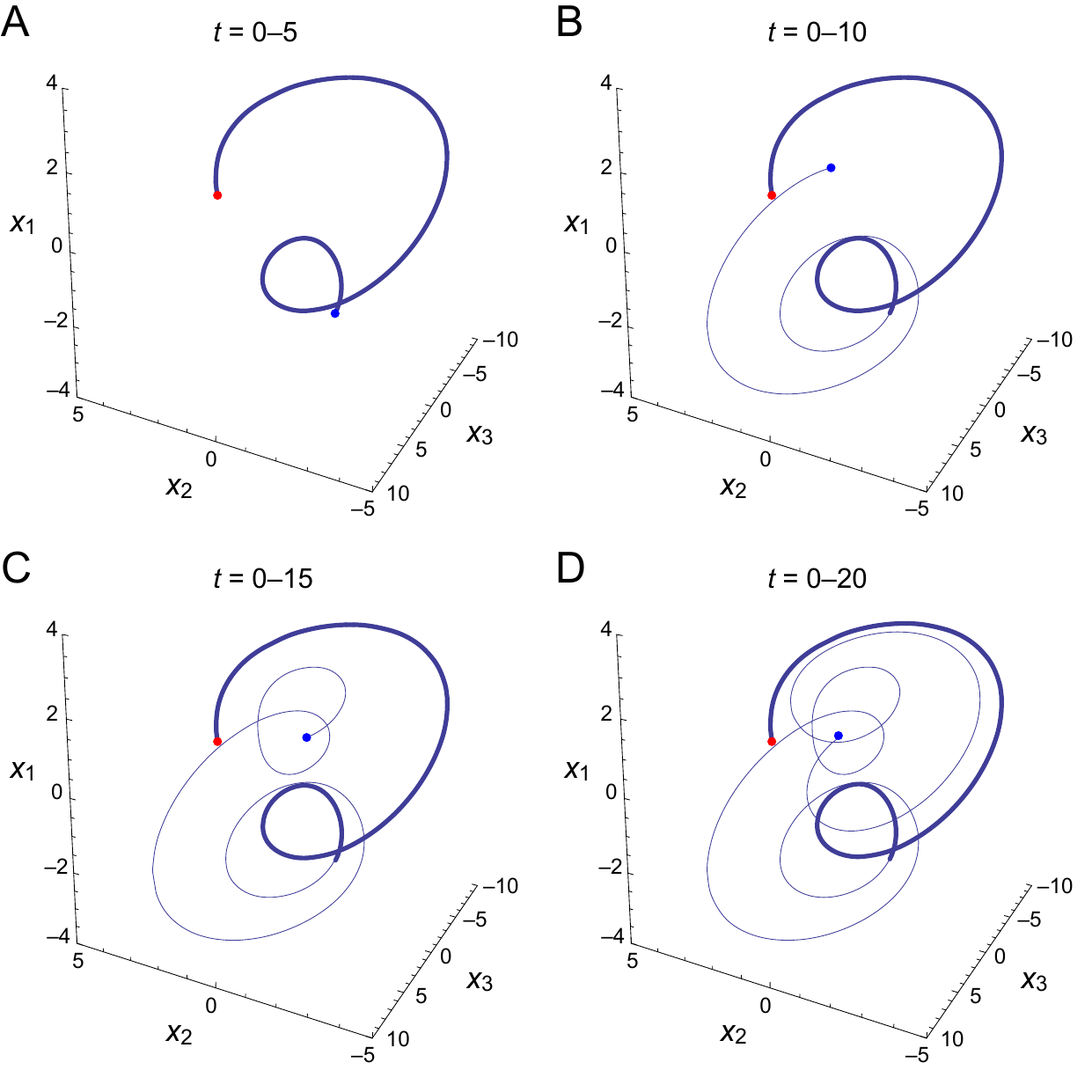}
\par\end{centering}

\caption{\label{fig:3}Deterministic dynamics of a physical system whose time
evolution is governed by the ordinary differential equation \eqref{eq:1}.
Each of the four simulation runs starts from the same initial state
$S(0)$ with $x_{1}(0)=1$, $x_{2}(0)=2$ and $x_{3}(0)=3$, but lasts
for different period of time $t=$0--5 in (A), $t=$0--10 in (B),
$t=$0--15 in (C) or $t=$0--20 in (D). The characteristic feature
of deterministic dynamics is that no matter how many runs are performed with the same initial state,
the initial segment of the trajectory for the period $t=$0--5 (thick purple line) will
always remain the same. We have performed four simulations with different durations because if we had simulated the same period of time $t=$0--5 consecutively four times, we would have ended with four
identical panels in the composite figure without any evidence that these were four different simulations.
The initial state is shown with a red point, whereas the final state is shown with a blue point.
The units of $t$ and $x_1$ are arbitrary, whereas the units of $x_2$ and $x_3$ are fixed as corresponding rates of change by \eqref{eq:3} and \eqref{eq:4}.}
\end{figure}

The mathematical properties of ordinary differential equations may
look innocent, but create serious difficulties once we start bringing
forward putative mind--brain models.

\subsection{Classical functionalism implies causally impotent consciousness}
\label{sub:31}

\begin{definition}
(Functionalism)
\emph{Functionalism} is the philosophical stance that the observable brain
\emph{produces} the unobservable conscious experiences. The main
feature is that the brain $\Phi$ produces the mind $\Psi$, namely
$\Phi\rightarrow\Psi$, but the brain $\Phi$ and the mind $\Psi$
are mathematically distinct entities, $\Phi\neq\Psi$.
The symbol $\rightarrow$ is used to indicate \emph{functional production}.
Here, the word \emph{produces} can be replaced with other
synonyms without changing the definition of functionalism, e.g.,
the brain \emph{gives rise} to the mind, the brain
\emph{generates} the mind or the brain \emph{creates} the mind.
The same meaning can also be expressed as: the mind \emph{emerges
from} the brain, the mind \emph{originates from} the brain or the
mind \emph{is constructed by} the brain.
\end{definition}

The classical aspect of brain modeling comes from the requirement
that all physical observables of the brain are governed by ordinary
differential equations (ODEs). This means that given the initial state
of the brain $\Phi(0)$ we can compute deterministically (i.e., solve numerically using a computer program) the future
state of the brain $\Phi(t)$ for any time $t>0$. Furthermore, because
the mind $\Psi$ is unobservable, it is not present as a variable
in the set of ODEs that govern brain dynamics.
The combination of classical physics and functionalism inevitably
implies that the mind is causally impotent and cannot affect anything
inside the physical world \cite{Georgiev2017}.

\begin{theorem}
\label{thm:1}
Classical functionalism implies that the conscious mind lacks causal potency and is unable to affect the physical world.
\end{theorem}
\begin{proof}
The main premises of classical functionalism can be summarized as
follows:

Premise 1. The act of functional production of the mind by the brain,
$\Phi(t)\rightarrow\Psi(t)$, entails an infinite list of productions
at each time point $t$: $\Phi(0)\rightarrow\Psi(0)$, $\Phi(t_1)\rightarrow\Psi(t_1)$,
$\Phi(t_2)\rightarrow\Psi(t_2)$, $\ldots$, $\Phi(t_n)\rightarrow\Psi(t_n)$.

Premise 2. The physical states of the brain $\Phi(t)$ and the surrounding
world $W(t)$ are governed by an explicitly given system of ordinary
differential equations (ODEs) in which only physical observables of
the brain and the world are present. The interaction between the brain and the surrounding world is reflected in the fact that the system of ODEs may not be separable in terms of brain and world variables.

From the second premise, which gives explicitly the system of physical
equations, we can use the initial states of the brain $\Phi(0)$ and
the world $W(0)$ to deterministically compute the future states of
the brain $\Phi(t)$ and the world $W(t)$ for any future time $t>0$.
These future states $\Phi(t)$ and $W(t)$ will remain exactly the
same regardless of what conscious experiences $\Psi(0),\Psi(t_1),\Psi(t_2),\ldots$
are produced by the brain, due to the fact that the conscious experiences
do not enter in the physical equations. Moreover, there will be no
difference whatsoever between the two cases in which any conscious
experiences $\Psi(0),\Psi(t_1),\Psi(t_2),\ldots$ are produced (case 1)
or not produced (case 2), since the computation of future states $\Phi(t)$
and $W(t)$ does not require knowledge of $\Psi(0),\Psi(t_1),\Psi(t_2),\ldots$
The causal effect of the conscious experiences is determined by subtraction
of the brain and world dynamics obtained in the absence of conscious
experiences (case 2) from the dynamics obtained in the presence of
conscious experiences (case 1). Because both dynamics are the same,
after the subtraction we are left with zero causal effect of conscious experiences
\begin{gather}
\Phi_\textrm{case~1}(t)-\Phi_\textrm{case~2}(t) = 0 \\
W_\textrm{case~1}(t)-W_\textrm{case~2}(t) = 0
\end{gather}
Therefore, the conscious experiences $\Psi(0),\Psi(t_1),\Psi(t_2),\ldots$ that are
produced by the classical brain cannot affect anything in the brain
or the surrounding world.

The proof of the theorem remains unaffected even if we assume that
the brain has to perform some ``function'' in order to generate
conscious experiences. Because the brain ``function'' does
not correspond to a single brain state at a single time point, but
corresponds to a sequence of brain states that form a trajectory from
$\Phi(t_{0})$ to $\Phi(t_{1})$ to generate the conscious experience
$\Psi(0)$, we only need to discretize the first premise as follows:
$[\Phi(t_{0})\sim\Phi(t_{1})]\rightarrow\Psi(0)$,
$[\Phi(t_{1})\sim\Phi(t_{2})]\rightarrow\Psi(t_1)$,
$[\Phi(t_{2})\sim\Phi(t_{3})]\rightarrow\Psi(t_2)$, $\ldots$, etc.,
where the symbol $\sim$ indicates the unique trajectory with the
given initial and final state. Solving the system of ODEs uses only
the second premise, which again reproduces the lack of causal potency
of the functionally generated mind upon the brain or the surrounding
world.
\end{proof}

The implications of Theorem~\ref{thm:1} are that neither human consciousness nor animal consciousness could have evolved through natural selection in classical functionalism \cite{Georgiev2017}.
Furthermore, if conscious experiences were unable to affect anything in the physical world, then their presence would have been utterly
meaningless \cite{James1879}. Since the human mind has evolved naturally and has left
cultural artifacts such as tools \cite{Zupancich2016}, musical instruments \cite{Conard2009b}, hand-carved
statuettes \cite{Conard2009a} and wall paintings in prehistoric caves \cite{Burkitt1921,Chauvet1996,Chauvet2001,Aubert2014,Hoffmann2018,Pike2012,Lewis2002,Lewis2003,Gross2020}, it follows that classical functionalism has to be false.\\

Some philosophers have attempted to express the mathematical theory
behind ordinary differential equations (ODEs) as a list of independent
verbal postulates and have misleadingly identified the contents of
Theorem~\ref{thm:1} as the ``problem of mental causation''
\cite{Kim1997,Kim1998}. The main philosophical goal was to isolate
the verbal postulate that is ``wrong'' and fix it. Ultimately,
it was concluded that the ``causal closure'' of the physical
world is the main culprit leading to causally impotent consciousness,
hence it is mandatory to resort to ``reductionism'' if we
were to have a theory of causally potent consciousness \cite{Kim2003}.
Although we consider that the overall move towards reductionism is
on the right track, there are several important inaccuracies that
make this previous philosophical work wanting:

First, the mathematical theory of ordinary differential equations
(ODEs) cannot be split into independent verbal postulates from which
one is allowed to chose from. Instead, the theory of ODEs is build
upon prerequisite mathematical concepts introduced in ordinary calculus
such as mathematical \emph{functions}, \emph{series}, \emph{limits}, \emph{derivatives} and \emph{integrals}
\cite{Strang1991}. The resulting mathematical properties of the solutions
of ODEs, including \emph{existence}, \emph{determinism} and \emph{uniqueness} given exact
initial condition, are inherited all together and cannot be dropped
selectively as one pleases.

Second, a major culprit leading to causally impotent consciousness
is the fact that the conscious states \emph{are not present inside} the system
of ODEs. In fact, whether the mind is labeled as ``physical''
or ``non-physical'' is irrelevant for its ``closure''
from the physical world. Even if the mind is defined to be ``physical'',
it will not be able to affect the brain or the surrounding world provided
that the mind is absent from the ODEs that govern the time dynamics
of the brain and the surrounding world.

Third, the \emph{determinism} of ODEs is indispensable for the proof
of Theorem~\ref{thm:1}. The subtraction procedure between solutions
obtained with the same initial conditions in different simulation
runs returns zero value only for ODEs, but not for stochastic differential
equations (SDEs) as we shall see in Section~\ref{sec:4}.

Fourth, the contents of Theorem~\ref{thm:1} is not itself a ``problem'',
but rather it is an ``indicator'' that at least one or maybe
both of the listed premises are false in the actual world in which
we live in, where evolution of human consciousness is possible through
natural selection.\\

The importance of Theorem~\ref{thm:1} is that it provides concrete guidelines on how a putative physical
theory of causally potent consciousness should look like. In particular,
from the constructed proof of the theorem, it is clear that the conscious
experiences $\Psi(0),\Psi(t_1),\Psi(t_2),\ldots,\Psi(t_n)$ \emph{need to be present inside}
the physical equations in order to affect the future dynamics of the
brain and the physical world. This endorses a form of ``reductionism'',
albeit not necessarily a classical one as we shall see next.

\subsection{Classical reductionism implies lack of free will}
\label{sub:32}

\begin{definition}
(Reductionism)
\emph{Reductionism} is the philosophical stance that the brain (or a part
of the brain) \emph{is identical with} the mind, $\Phi=\Psi$.
We have already discussed clinical evidence that the seat of human consciousness
is located in the \emph{brain cortex} (cf. Section~\ref{sec:2}).
Pinpointing the exact brain area hosting consciousness, however, may
not be relevant for the subsequent discussion because cortical pyramidal
neurons are comprised from the same basic chemical elements as other
cells in the body, hence the experiential aspect is immediately attributed to all living matter.
Since the logical \emph{identity} relation goes in both ways, it implies
not only that human conscious experiences are built up from chemical
atoms, but also that chemical atoms outside of the human brain \emph{are}
conscious experiences whose only distinction is that they are not ``our'' conscious experiences.
In other words, reductionism endorses a form of ``panpsychism'' or ``panexperientialism'', according to which the physical world is comprised from conscious experiences.
These experiences can be attributed to multiple minds, thereby avoiding the depressing thesis of ``solipsism,'' according to which our mind is the only thing that exists in the universe.
\end{definition}

Expressed as a percentage from the total mass,
the elemental composition of the brain is 74.4\% oxygen~(O), 12.1\%
carbon~(C), 10.6\% hydrogen~(H), 1.8\% nitrogen~(N), 0.4\% phosphorus~(P),
0.1\% sulfur~(S) and 0.6\% other elements (Na, K, Mg, Cl, etc.).
Our estimates for the brain elements are based on the chemical composition
of the brain including 78.9\% water, 11.3\% lipids, 9.0\% proteins,
0.2\% nucleic acids and 0.6\% electrolytes \cite{OBrien1965,Winick1968,Friede1971}.
The chemical composition of the brain contains fewer chemical elements compared to the complete periodic table of chemical elements.
This leaves the theoretical possibility to identify conscious
experiences with only those chemical elements that appear in the brain,
while leaving other elements that occur naturally in the surrounding
world as lacking any experiences. Here, we have not entertained such
a partial attribution of experiences to only part of the world because
all chemical elements are ultimately composed from protons, neutrons and electrons.
If~experiences are attributed to all
elementary particles, the result will be exactly \emph{panexperientialism}.

The reductive claim expressed by the \emph{mind--brain identity}, $\Phi=\Psi$,
is that the two distinct labels ``$\Phi$'' and ``$\Psi$''
refer to the same thing in physical reality. For example, when we
refer to the real person Alice, we may use either the personal name
``Alice'' or words like ``she'' or ``herself''.
Even though the three words, ``Alice'', ``she'' and
``herself'', are different, they all mean the same thing,
which is that particular individual named Alice. The distinction between
the label ``Alice'' and the real person Alice is the same
as between the \emph{map} and the \emph{territory} \cite{Korzybski1994}. In other
words, the mind--brain identity stipulates that when we talk about
``pyramidal neurons firing electric spikes in the brain cortex'',
we literally talk about conscious experiences that exist in the real
world. Or to put it differently, the reductionism denies the philosophical
stance of ``naive realism'' according to which the phrase ``pyramidal
neurons firing electric spikes in the brain cortex'' refers
to existing insentient pyramidal neurons firing electric spikes in the brain cortex.
If what exists in the world as \emph{territory} is only the conscious mind (experience), then
the ``brain'' is just a label on the \emph{map} of our scientific theory that refers to the existing mind.
In contrast, the mind--brain problem is unavoidable in ``naive realism'' where the ``brain'' refers to an existing
brain composed on insentient matter that is then expected to somehow generate
the conscious experiences (sentience) that comprise the conscious mind.
Having clarified the meaning of the mind--brain identity thesis, we
are ready for its implications in classical physics:

\begin{theorem}
\label{thm:2}
Classical reductionism implies that the conscious mind
is causally potent and affects the physical world. Such conscious
mind, however, lacks free will because it cannot choose its dynamics.
\end{theorem}
\begin{proof}
Provided with an initial state of the brain $\Phi(0)$ and the world $W(0)$, it is always
possible to deterministically compute the future states of the brain
$\Phi(t)$ and the world $W(t)$ for any future $t>0$. However, since the brain and the
mind are identical, $\Phi(t)=\Psi(t)$, for all $t\geq 0$, it follows
trivially that the initial state of the mind $\Psi(0)$ affects the
future state of the brain $\Phi(t)$. The brain is part of the world,
which means that the mind affects the world. This is not surprising
because from the mind--brain identity, the nature of the physical world becomes essentially mental.
Because consciousness is present in the physical equations,
the subtraction argument based on simulation of brain dynamics in the presence or absence of
conscious experiences cannot be applied to reductionism due to the fact that the
\emph{identity of an entity with itself} cannot be turned off.
By logical necessity, it is always true that $\Phi(t)=\Phi(t)$ and it is always
false that $\Phi(t)\neq\Phi(t)$. Therefore, if $\Phi(t)$ is a mental
state, then its mental nature cannot be turned off without negating
$\Phi(t)$ itself. This concludes the proof of the first part of the theorem.
The second part, which concerns the lack of free will, follows directly from the deterministic
dynamics of ODEs.
\end{proof}

The main advantage of reductionism is that the conscious mind is present
in the physical equations that govern the future time dynamics of
the physical world. The resulting panexperientialism then provides
a mental substrate from which more complex and elaborate minds could
evolve through natural selection. In the words of the psychologist William
James, the natural evolution of the human mind requires the existence
of ``mind dust'' in nature \cite{James1890}. We agree with
that conclusion, but we have formulated it in terms of physical reductionism
and presence of conscious experiences in the system of ODEs that govern
the behavior of the physical world.

Still, there is an unsettling inconsistency between classical
physics and panexperientialism, because classical physics postulates
that the brain is ``observable'', whereas the phenomenological,
qualitative nature of conscious experiences is ``unobservable''
\cite{Georgiev2020a,Georgiev2020b}. To fix this serious problem,
one would need to admit that not all aspects of physical reality are
observable, hence classical physics needs to be repaired in some form
or another. Fortunately, modern physicists have already found a better
replacement for classical physics, which is provided by quantum physics
\cite{Dirac1967,Susskind2014,Bhaumik2019}. Framing the mind--brain problem in
the context of quantum physics, not only solves the ``unobservability''
and the ``causal potency'' of consciousness, but also introduces
``free will'' as we shall see next.

\section{Causal potency of consciousness in quantum physics}
\label{sec:4}

Quantum physics conflicts conceptually with classical physics in a remarkable way.
Because quantum systems possess the ability to accomplish physical tasks that are classically impossible,
their ``quantumness'' or ``quantum nature'' is a valuable \emph{physical resource} that is worth having \cite{Georgiev2022b,Georgiev2022c}.
Noteworthy, in the quantum world, there is a fundamental dichotomy between ``existence''
and ``observability'' because \emph{what exists} is different
from \emph{what can be observed} \cite{Georgiev2017,Georgiev2020a,Georgiev2020b}.
Mathematically, this difference is expressed in the fact that \emph{quantum
states} are \emph{vectors}~$\vert\Psi\rangle$ in Hilbert space~$\mathcal{H}$,
whereas \emph{quantum observables} $\hat{A}$ are \emph{operators}
on the Hilbert space~$\mathcal{H}$ \cite{Dirac1967,Susskind2014,vonNeumann1932,vonNeumann1955,Bhaumik2020}.
For example, the general quantum
state of a spin-1 particle can be written as
\begin{equation}
\vert\Psi\rangle=\left(\begin{array}{c}
\alpha_{1}\\
\alpha_{2}\\
\alpha_{3}
\end{array}\right)=\alpha_{1}\vert\uparrow_{z}\rangle+\alpha_{2}\vert \bigcirc_{z}\rangle+\alpha_{3}\vert \downarrow_{z}\rangle
\end{equation}
where $\vert \uparrow_{z}\rangle$, $\vert \bigcirc_{z}\rangle$ and $\vert \downarrow_{z}\rangle$
are the eigenvectors of the $z$-component of the spin-1 observable
\begin{equation}
\hat{\sigma}_{z}=\left(\begin{array}{ccc}
1 & 0 & 0\\
0 & 0 & 0\\
0 & 0 & -1
\end{array}\right)=1\vert \uparrow_{z}\rangle\langle\uparrow_{z}\vert +0\vert \bigcirc_{z}\rangle\langle\bigcirc_{z}\vert -1\vert \downarrow_{z}\rangle\langle\downarrow_{z}\vert 
\end{equation}
with corresponding eigenvalues $1$, 0 and $-1$. The relationship
between eigenvectors and eigenvalues of quantum observables and quantum
measurements of quantum states is given by the Born rule \cite{Born1955,Vaidman2020} as follows:
Suppose that we have the quantum state $\vert \Psi\rangle$ and we measure
the quantum observable $\hat{\sigma}_{z}$. Then the quantum system
chooses indeterministically the outcomes from the eigenvalues of $\hat{\sigma}_{z}$
and at the same performs a quantum jump into the corresponding eigenvector
of $\hat{\sigma}_{z}$. The probabilities for the different choices
are computed from the quantum probability amplitudes that define the
quantum state $\vert \Psi\rangle$, namely, the quantum system produces
the eigenvalue outcome $1$ and jumps into the eigenstate $\vert \uparrow_{z}\rangle$
with probability~$\vert \alpha_{1}\vert ^{2}$, produces the eigenvalue outcome
$0$ and jumps into the eigenstate $\vert \bigcirc_{z}\rangle$ with probability~$\vert \alpha_{2}\vert ^{2}$,
or produces the eigenvalue outcome $-1$ and
jumps into the eigenstate $\vert \downarrow_{z}\rangle$ with probability~$\vert \alpha_{3}\vert ^{2}$.
Symbolically, we write the probabilities for
the alternative quantum jumps as
\begin{eqnarray}
\textrm{prob}\left(\vert \Psi\rangle\hookrightarrow\vert \uparrow_{z}\rangle\right) & = & \langle\Psi\vert \uparrow_{z}\rangle\langle\uparrow_{z}\vert \Psi\rangle=\vert \alpha_{1}\vert ^{2}\\
\textrm{prob}\left(\vert \Psi\rangle\hookrightarrow\vert \bigcirc_{z}\rangle\right) & =& \langle\Psi\vert \bigcirc_{z}\rangle\langle\bigcirc_{z}\vert \Psi\rangle=\vert \alpha_{2}\vert ^{2}\\
\textrm{prob}\left(\vert \Psi\rangle\hookrightarrow\vert \downarrow_{z}\rangle\right) & =& \langle\Psi\vert \downarrow_{z}\rangle\langle\downarrow_{z}\vert \Psi\rangle=\vert \alpha_{3}\vert ^{2}
\end{eqnarray}

Quantum indeterminism is a characteristic feature of quantum systems
because for the description of the outcomes of a sequence of quantum
measurements performed on a dynamic quantum system, one needs to introduce
a stochastic process $\Gamma(t)$ and solve \emph{stochastic differential
equations} (SDEs) \cite{Gudder1979}. Because the \emph{stochastic calculus} that is
needed to solve SDEs was developed by the Japanese mathematician Kiyosi
It\^{o} \cite{Ito1944,Ito1950,Ito1975,Ito1984}, it is often referred to as \emph{It\^{o} calculus} \cite{Ikeda1996}.

\begin{example}
(Stochastic dynamics) \label{exa:2}To illustrate the mathematical
properties of quantum stochastic dynamics, consider the jerky
Lorenz-like system into which is injected the quantum stochastic process~$\Gamma(t)$ 
\begin{equation}
\left(\frac{d}{dt}\right)^{3}x=-\left(\frac{d}{dt}\right)^{2}x-4\frac{d}{dt}x+5x-x^{3}+\Gamma
\label{eq:9}
\end{equation}
where $\Gamma(t)$ is a continuous function obtained through linear
interpolation of a sequence of unbiased choices from the set $\{1,0,-1\}$
performed at unit time intervals. The remarkable feature of stochastic
dynamics is that every simulation run produces almost surely a different
result (Fig.~\ref{fig:4}). This is a direct consequence from the
fact that the stochastic process $\Gamma(t)$ produces different sequences
of chosen outcomes for different simulation runs. If the simulation
is run for $t$ units of time, the stochastic process $\Gamma(t)$
will involve $t+1$ unbiased choices. Therefore, the probability to
produce two identical simulation runs is $\textrm{prob}(\textrm{run~1}=\textrm{run~2})=\left(\frac{1}{3}\right)^{t+1}$,
which approaches zero in the limit $t \to \infty$. It is in this sense that different
stochastic runs produce ``almost surely'' different results.
\end{example}
\begin{figure}
\begin{centering}
\includegraphics[width=\textwidth]{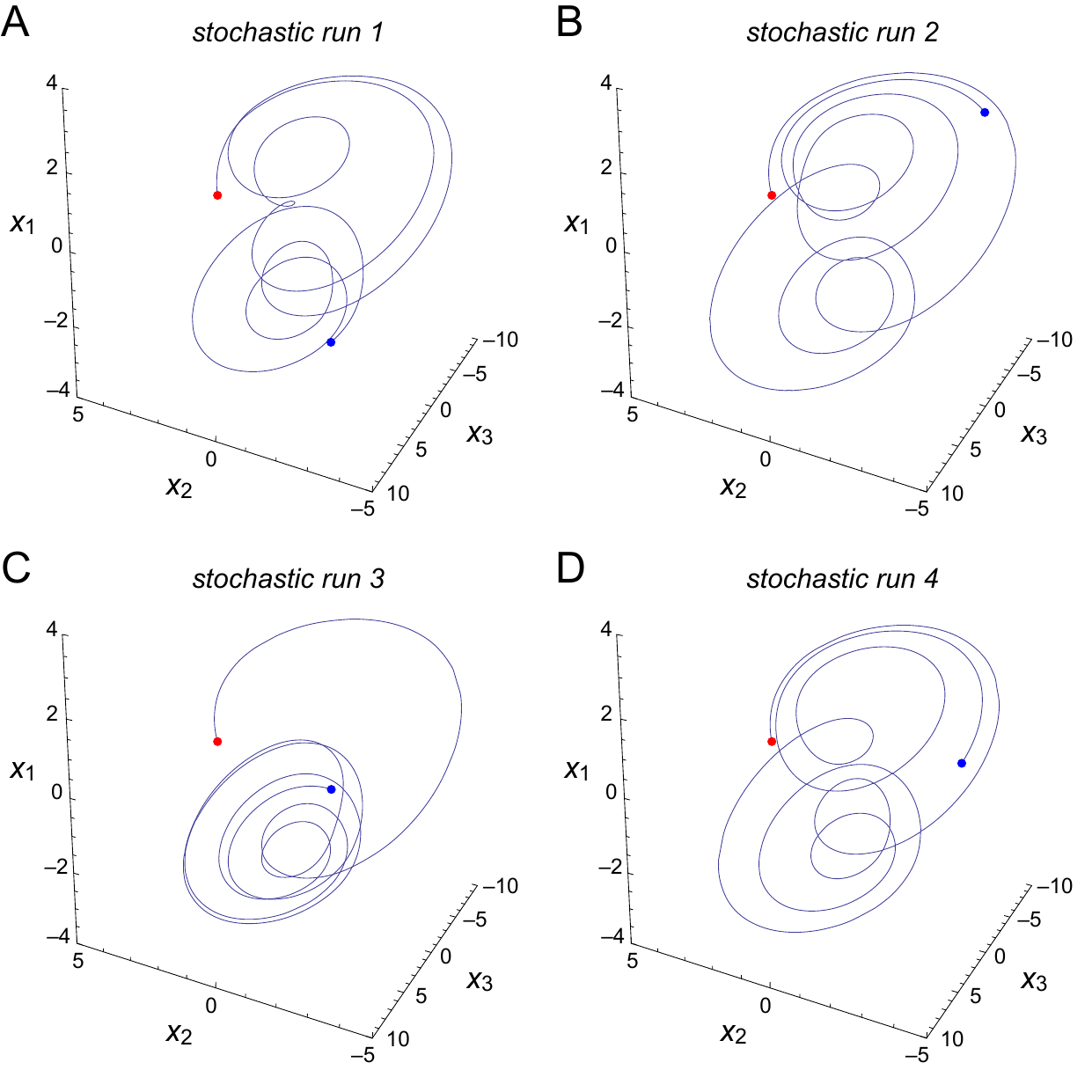}
\par\end{centering}

\caption{\label{fig:4}Stochastic dynamics of a physical system whose time
evolution is governed by the stochastic differential equation \eqref{eq:9}.
Each of the four simulation runs (A--D) starts from the same initial
state $S(0)$ with $x_{1}(0)=1$, $x_{2}(0)=2$ and $x_{3}(0)=3$
and lasts for exactly the same period of time $t=$0--20. The characteristic
feature of stochastic dynamics is that for each run the dynamic trajectory
is almost surely going to be different as the outcomes of each stochastic
choice are drawn from a certain probability distribution. The initial
state is shown with a red point, whereas the final state is shown
with a blue point.
The units of $t$ and $x_1$ are arbitrary, whereas the units of $x_2$ and $x_3$ are fixed as corresponding rates of change by \eqref{eq:3} and \eqref{eq:4}.}
\end{figure}

The theory of stochastic differential equations (SDEs) is very rich
and includes the theory of ordinary differential equations (ODEs)
as a special case. There are two ways that one can obtain deterministic
dynamics. First option is to consider the expectation values of the
quantum measurements performed on the quantum system \cite{Georgiev2019,Georgiev2020d,Georgiev2022,Georgiev2022b}.
This means that one needs to collect a sufficiently large sample of
individual stochastic trajectories and then compute a single average
trajectory. The computation of such an average trajectory may or may
not be useful and its physical interpretation may or may not be meaningful due to potential presence of outliers.
For example, the average of 1 billionaire and 999 poor people will
be a ``millionaire''. Learning that the ``average''
person from the group of 1000 people is a ``millionaire'',
however, is useless and utterly misleading because the distribution
of the sample consists of 99.9\% of poor people. Thus, working with
expectation values and interpreting those expectation values requires
good understanding of statistics \cite{Huff1993} and quantum foundations
\cite{Gudder1988}. Second option to obtain deterministic dynamics
is to consider highly biased probability distributions with zero variance.
In Example~\ref{exa:2}, we have considered unbiased stochastic process $\Gamma(t)$,
which makes the resulting simulated trajectories equally likely. However,
a uniform probability distribution can be continuously transformed
into a highly biased probability distribution that is narrowly peaked
onto a single outcome with zero variance. Quantum measurement theory \cite{vonNeumann1955,Busch1996} allows for physical
realization of the full spectrum of quantum probability distributions
from a completely uniform distribution to a highly nonuniform distribution
consisting of a single narrow peak onto a single outcome. This means
that the behavior of quantum systems depends critically on the measurement
context. For some quantum measurement contexts, the resulting dynamic
trajectory may appear to be random, whereas for other quantum measurement
contexts the resulting dynamic trajectory may appear to be deterministic \cite{Georgiev2021} (see also Section~\ref{S5}).
That is why quantum physics is indispensable for the proper understanding
of consciousness and free will.
Furthermore, the causal potency of consciousness is no longer threatened in the quantum world,
as we shall demonstrate next.

\subsection{Quantum functionalism does not guarantee causally potent consciousness}
\label{sub:41}

Quantum functionalism understood as the \emph{quantum brain state}
$\vert \Phi\rangle$ producing the \emph{conscious mind} $\Psi$ is not in itself directly
incompatible with causally potent consciousness. If the quantum
stochastic dynamics of the brain state $\vert \Phi(t)\rangle$ is governed
by a system of SDEs, it would follow that different simulation runs
produce different results, hence the subtraction of two different
stochastic runs will not be zero. This leaves room for the introduction of conscious
action that is causally potent.

It should be noted that the brain is an open quantum system interacting with its physical environment, because the brain inputs sensory information and outputs motor information (Fig.~\ref{fig:1}). Nevertheless, we prefer writing the quantum brain state with the use of the brain quantum state vector~$\vert \Phi\rangle$, rather than the brain density matrix~$\hat{\rho}$, because we model ``genuine'' quantum stochastic dynamics that is due to the presence of objective wave function collapses \cite{Ghirardi1986,Ghirardi1990,Bassi2003} that lead to actualization of single measurement results with intermittent disentanglement of the brain and its physical environment \cite{Georgiev2017} (see also Section~\ref{S8}).
In fact, it can be shown that ``no collapse'' models of quantum mechanics, in which the whole universe as a closed system is described only by the Schr\"{o}dinger equation, are no different than the classical models based on ordinary differential equations (ODEs). In particular, the Schr\"{o}dinger equation is an ODE and the partial trace operation is perfectly deterministic procedure, which would imply that Theorems~\ref{thm:1} and \ref{thm:2} apply to such ``no collapse'' quantum theory of consciousness. Extensive criticism of epiphenomenal consciousness in ``no collapse'' models of quantum mechanics has already been presented elsewhere \cite{Georgiev2017} and will not be repeated here. Instead, we would like to focus on the resolution of the problem with mental causation provided by genuine stochastic ``quantum jumps'' that are mathematical representation of the objective physical wave function collapses in ``dynamical collapse'' models of quantum mechanics.

\begin{theorem}
\label{thm:3}
Quantum functionalism does not guarantee causally potent consciousness, but leaves room for upgrading the theory to a form of interactive mind--brain dualism.
\end{theorem}
\begin{proof}
The main premises of quantum functionalism can be summarized as follows:

Premise 1. The act of functional production of the mind by the brain,
$\vert \Phi(t)\rangle\rightarrow\Psi(t)$, entails an infinite list of
productions at each time point $t$: $\vert \Phi(0)\rangle\rightarrow\Psi(0)$,
$\vert \Phi(t_1)\rangle\rightarrow\Psi(t_1)$, $\vert \Phi(t_2)\rangle\rightarrow\Psi(t_2)$,
$\ldots$, $\vert \Phi(t_n)\rangle\rightarrow\Psi(t_n)$.

Premise 2. The quantum states of the brain $\vert \Phi(t)\rangle$ and
the surrounding world $\vert W(t)\rangle$ are governed by an explicitly
given system of stochastic differential equations (SDEs) in which
only physical observables of the brain and the world are present.

From the second premise, we can infer that for each simulation run
the quantum states of the brain $\vert \Phi(t)\rangle$ and the surrounding
world $\vert W(t)\rangle$ undergo sequences of quantum jumps, which we
will denote with the symbol $\hookrightarrow$ as follows:
\begin{eqnarray}
\vert \Phi(t)\rangle & = &\vert \Phi(0)\rangle\hookrightarrow\vert \Phi(t_1)\rangle\hookrightarrow\vert \Phi(t_2)\rangle\hookrightarrow\ldots\hookrightarrow\vert \Phi(t_{n})\rangle\hookrightarrow\ldots\\
\vert W(t)\rangle & = &\vert W(0)\rangle\hookrightarrow\vert W(1)\rangle\hookrightarrow\vert W(2)\rangle\hookrightarrow\ldots\hookrightarrow\vert W(t_{n})\rangle\hookrightarrow\ldots
\end{eqnarray}
The unitary quantum interaction between the brain and its environment will lead to production of quantum entangled clusters of neurons in the brain cortex that will disentangle to the corresponding product states $\vert \Phi(t)\rangle\otimes \vert W(t)\rangle$ at the instances $t_0, t_1, \ldots, t_n$ when the definite measurement results are actualized through ``dynamical collapses'' \cite{Georgiev2017} (see also Section~\ref{S8}).
All the unitary quantum dynamics resulting from the Schr\"{o}dinger equation, which is an ODE, will be dependent on the actual brain quantum Hamiltonian and is purposefully left implicit in the $\hookrightarrow$ symbols to prevent unnecessary distraction.
The \emph{conceptual highlight} in the above mathematical description is that, in general, performing a subtraction for two different stochastic
runs 1 and 2 will not produce a zero result
\begin{gather}
\vert \Phi(t)\rangle_{\textrm{run }1}-\vert \Phi(t)\rangle_{\textrm{run }2}\neq 0 \\
\vert W(t)\rangle_{\textrm{run }1}-\vert W(t)\rangle_{\textrm{run }2}\neq 0
\end{gather}
This non-zero difference does not have to be attributed to the action
of the conscious mind $\Psi(t)$, but if the quantum functionalism
wants to have a theory of consciousness that is consistent with natural
evolution, then it is possible to introduce as a postulate that the
conscious mind is the agent that chooses the particular outcomes for
the brain states at each quantum jump. The chosen brain states will
then affect the state of the surrounding world, and the resulting
theory will be a form of interactive mind--brain dualism.
\end{proof}

Despite that quantum physics is not incompatible with functionalism and causally potent consciousness, the very idea of ``production'' or ``emergence'' of consciousness is problematic for the following reasons.

First, the postulated emergence of consciousness is \emph{ad hoc}
and not different from the postulation of occurrence of ``miracles''.
For example, one may as well postulate that a flutter of fairies appeared
from the brain, then they performed a dance to welcome their fairy
queen, and finally they all decided what the next brain state should
be as an outcome of the current quantum jump. In other words, it is
quite unsettling that the brain did not have the capacity to perform
the choice of the quantum jump by itself without resort to any external
agency.

Second, even if the emergence of a conscious mind $\Psi_{i}$ is granted
for each quantum brain state $\vert \Phi_{i}\rangle$, it is not clear
what prevents the possibility of paranormal action? For example, how
is it possible that the conscious mind $\Psi_{i}$ recognizes that
it can act upon the quantum brain state $\vert \Phi_{i}\rangle$ but not
on another present quantum brain state $\vert \Phi_{j}\rangle$? Or to
put it in more familiar terms, what is the physical mechanism that
prevents Alice's consciousness to act upon Bob's brain and vice versa?

Third, even more severe objection to functionalism is the fact that
it is easy for one to falsify it empirically. For example, to postulate
that our conscious mind chooses what our brain state should be, implies
that we are knowingly selecting which neuron in our brain will be
firing and which neuron will remain silent, which ion channel is be
open for the passage of electric current and which ion channel will
remain closed. Introspectively, we can verify that we have no idea
what brain state we are choosing. What is more, no human being on
Earth has the slightest idea of what the actual chemical composition
of their own brain is. Therefore, it seems that we are not knowingly
choosing our brain states.

Fortunately, a safe way out from all these problems is provided by
\emph{quantum reductionism}.

\subsection{Quantum reductionism guarantees causally potent consciousness and free will}
\label{sub:42}

Quantum reductionism identifies the quantum state of the brain $\vert \Phi\rangle$
with the conscious mind $\vert \Psi\rangle$ at all times, namely, $\vert \Phi(t)\rangle=\vert \Psi(t)\rangle$.
Because the identity of a thing with itself
cannot be logically turned off, it would then follow that all quantum
states in the quantum world are comprised of conscious experiences.
Thus, the resulting \emph{quantum panexperientialism} provides a mental fabric
for physical reality in which complex minds could evolve naturally
from simpler minds \cite{Georgiev2017}.

\begin{theorem}
\label{thm:4}
Quantum reductionism implies that the conscious mind
is causally potent and affects the physical world. Such conscious
mind possesses free will because it is able to choose among future courses of action.
\end{theorem}
\begin{proof}
Since the quantum brain state and the conscious mind are identical,
$\vert \Phi(t)\rangle=\vert \Psi(t)\rangle$, for all $t$, it follows
trivially that the initial state of the mind $\vert \Psi(0)\rangle$ affects
the future quantum state of the brain $\vert \Phi(t)\rangle$. The brain
is part of the world, which means that the mind affects the world.
The possession of free will follows from the stochastic dynamics.
A quantitative measure for the amount of manifested free will is provided
by the expected information gain from learning the actual sequence
of choices made by the conscious mind (for~details, see Sections~\ref{S5} and \ref{S6}).
\end{proof}

One of the advantages of quantum reductionism is that because the
quantum brain--mind identity relation goes in both ways, it would
follow that the conscious mind has all of the properties of a quantum
state and satisfies the axioms of a vector in Hilbert space. To highlight
this fact, we no longer use the bare symbol ``$\Psi$'' for
the conscious mind, but rather insert it inside a ket $\vert \Psi\rangle$
following Dirac's bra-ket notation \cite{Dirac1939,Dirac1967}. This is highly
informative from a theoretical perspective because one becomes equipped
with quantum information theoretic \emph{no-go theorems} that can
be applied to consciousness in order to determine some of its physical properties \cite{Georgiev2017,Melkikh2019}.
For example, the quantum state vector of quantum physical systems is not observable due to a theorem by Busch \cite{Busch1997}, which in turn explains why the conscious
experiences are not observable \cite{Georgiev2017,Georgiev2020a,Georgiev2020b}.
Deriving the ``unobservability'' of conscious experiences
as a prediction from the physical theory of consciousness is a remarkable
achievement, especially when compared to classical reductionism, which
is ripped by the internal inconsistency between the \emph{unobservable mind}
and the \emph{observable brain}.

Another achievement of quantum reductionism is the ability to explain
the physical difference between the unobservable mind and the observable
brain (Fig.~\ref{fig:5}). In a quantum world, what can be observed
is different from what exists in the form of a quantum state $\vert \Psi\rangle$.
This physical difference is reflected in the mathematical formalism
of quantum theory where quantum observables are described by operators
$\hat{A}$ on Hilbert space~$\mathcal{H}$. The observed outcomes
in quantum measurements are the eigenvalues of the measured quantum
brain observables. Because the eigenvalues are just numbers, these
can be represented by bits of classical information and stored on
a digital file as a string of 0s and 1s. This classical information
is the ``observable brain'' and can be communicated to multiple
external observers. For example, a microscopic picture of the anatomical
organ (brain) that is inside your skull shows a neural network of
neurons \cite{DeFelipe2002}. This neuronal picture is different from the conscious experiences
that exists in reality and is compatible with the fact that we do
not introspectively perceive ourselves as a collection of neurons.
In other words, the picture of neurons is not the conscious mind itself,
but what the conscious mind looks like from a third-person point of
view. The same picture of the observable brain can be copied, multiplied, communicated to
and simultaneously studied by multiple neuroscientists. The amount
of classical bits of information that can be obtained while observing
the quantum state of the brain $\vert \Psi\rangle$ is bound by Holevo's theorem in
quantum information theory \cite{Holevo1973,Georgiev2020a,Georgiev2020b}.

Clear conceptual distinction between the ``unobservable'' mind and the ``observable'' brain is preserved by avoiding the usage of a brain density matrix~$\hat{\rho}$, which is a \emph{quantum observable}. For a pure brain state, i.e., $\textrm{Tr}(\hat{\rho}^2)=1$ it follows that $\hat{\rho} = \vert \Psi\rangle\langle\Psi\vert $ hence the state can be equivalently written as a ket vector $\vert \Psi\rangle$.
Component parts of a composite quantum entangled system, however, do not have their own ket vectors, hence do not have their own minds, whereas they always have a reduced density matrix $\hat{\rho}$ that is not pure, i.e., $\textrm{Tr}(\hat{\rho}^2)<1$ \cite{Georgiev2017}. In other words, possessing a density matrix is not something special that is useful for \emph{demarcation of mind boundaries} because any collection of quantum particles is guaranteed to have a reduced density matrix. On the other hand, it is indeed something special for a collection of quantum particles to have their collective quantum state vector $\vert \Psi\rangle$ because in general not every collection of quantum particles is guaranteed to have a quantum state vector. What is more, the factorizability of $\vert \Psi\rangle$ is informative with regard to mind boundaries, namely, a nonfactorizable $\vert \Psi\rangle$ corresponds to a single mind, whereas a factorizable $\vert \Psi\rangle=\vert \psi_1\rangle \otimes \vert \psi_2\rangle\otimes \ldots \otimes \vert \psi_k\rangle$ corresponds to a collection of $k$ separate minds \cite{Georgiev2017,Georgiev2021b} (see also Section~\ref{S7}).

\begin{figure}
\begin{centering}
\includegraphics[width=\textwidth]{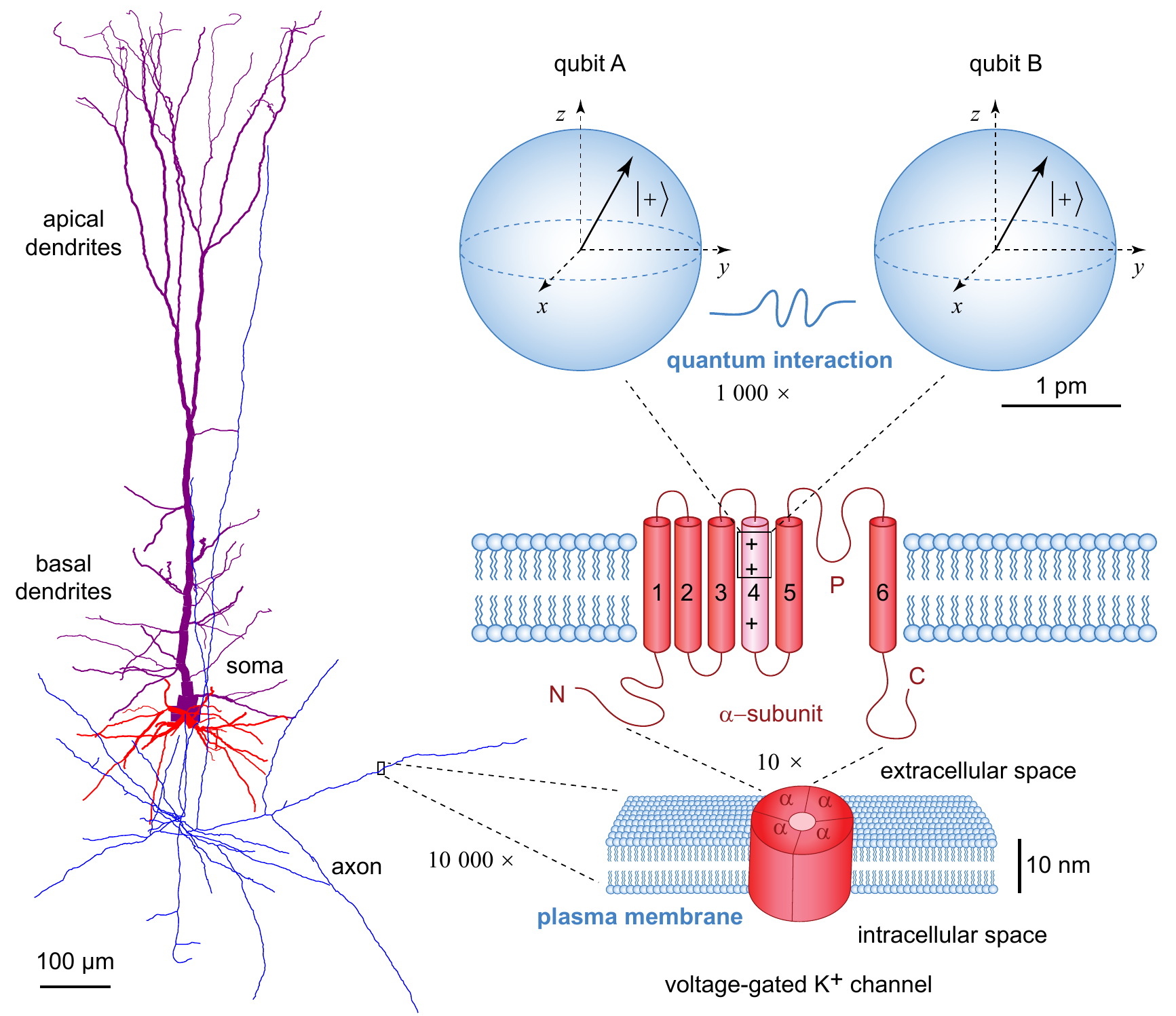}
\par\end{centering}

\caption{\label{fig:5}Different levels of organization of physical processes
within the central nervous system. At the microscopic scale, the brain
cortex is composed of neurons, which form neural networks. The morphology
of the rendered pyramidal neuron (\href{https://neuromorpho.org/neuron_info.jsp?neuron_id=NMO_77905}{NMO\_77905}) from layer~5 of rat
somatosensory cortex (http://NeuroMorpho.Org) reflects the functional
specialization of dendrites and axon for the input and output of electric
signals, respectively. At the nanoscale, the electric activity of
neurons is generated by voltage-gated ion channels, which are inserted
in the neuronal plasma membrane. As an example of ion channel is shown
a single voltage-gated K\protect\textsuperscript{+} channel composed
of four protein $\alpha$-subunits. Each subunit has six $\alpha$-helices
traversing the plasma membrane. The 4th~$\alpha$-helix is positively
charged and acts as voltage sensor. At the picoscale, individual elementary
electric charges within the protein voltage sensor could be modeled
as qubits represented by Bloch spheres. For the diameter of each qubit
is used the Compton wavelength of electron. Consecutive magnifications
from micrometer ($\mu$m) to picometer (pm) scale are indicated by
$\times$ symbol. Modified from Ref.~\citen{Georgiev2021b}.}
\end{figure}

\begin{example}
(Quantum reductionism forbids mind uploading) Within the quantum reductive approach,
profound insights into the nature of consciousness arise from the characterization
of every quantum state~$\vert \Psi\rangle$ as a quantum coherent superposition
of quantum probability amplitudes for potential future quantum events to occur.
The \emph{actualization} of one of these potential events occurs
during quantum measurements with their associated quantum jumps. For
example, consider the chemical composition of the voltage sensor of
voltage-gated K\textsuperscript{+} channel, which is built up from
carbon~(C), nitrogen~(N), oxygen~(O) and hydrogen~(H) atoms \cite{Kariev2014,Kariev2018,Kariev2021} (Fig.~\ref{fig:5}).
The positively charged hydrogen nuclei (protons) in the positively
charged Lysine or Arginine amino acid residues can be located in different
locations in space, the conformation of which determines whether the
voltage-gated K\textsuperscript{+} channel is in open or closed state \cite{Sansom2000,Catterall2010,Islas2016}.
Different conformations of Lysine or Arginine are realized by different
distributions of C, N, O and H atoms. Each distribution could be realized
with quantum probability given by the absolute square of the corresponding
quantum probability amplitudes. Importantly, these quantum probability
amplitudes are non-zero only for the mentioned C, N, O and H atoms,
but are zero for atoms of other chemical elements.
This leads to impossibility
of the quantum probability amplitudes to be separated from their physical
substrate. In modern science fiction scenarios, it is often imagined
that human consciousness could be uploaded onto silicon-based computer
chips \cite{Bamford2012,Choe2012,Depp2014,Cappuccio2017}.
Quantum reductionism, however, forbids physically such possibility. Indeed, let us imagine
for the sake of argument that the quantum state of the voltage-gated
K\textsuperscript{+} channel could be cloned onto the silicon chip.
Chemically, the silicon chip is comprised of silicon (Si) atoms. If
the quantum state of the silicon chip has perfectly become the quantum
state of the voltage-gated K\textsuperscript{+} channel, then it
would contain zero quantum probability amplitude for measurement of
silicon (Si) atoms. Consequently, after we interact the silicon chip
to measure it, we will observe the actualization of a voltage-gated
K\textsuperscript{+} channel. Similarly, if
we were able to clone the complete quantum state of the brain onto
a silicon chip, the silicon chip would turn into an organic brain
tissue upon observation. Logically, the quantum probability amplitude
for actualization of a brain does actualize a brain, but does not
actualize a silicon chip. That is why the quantum information contained
in the quantum probability amplitudes of a quantum state $\vert \Psi\rangle$
is fundamentally inseparable from its physical substrate. The quantum
probability amplitudes are for the actualization of something, and
this ``something'' is what we call the physical substrate
of the quantum state.
\end{example}

Quantum stochastic dynamics of the quantum brain state $\vert \Psi(t)\rangle$,
which is comprised of conscious experiences, solves at once both the
causal potency problem and the free will problem \cite{Georgiev2021}. If one performs
a series of quantum measurements upon the quantum brain state $\vert \Psi(t)\rangle$,
the result will be a sequence of actualized outcomes
\begin{equation}
\vert \Psi(t)\rangle=\vert \Psi(0)\rangle\hookrightarrow\vert \Psi(t_1)\rangle\hookrightarrow\vert \Psi(t_2)\rangle\hookrightarrow\ldots\hookrightarrow\vert \Psi(t_{n})\rangle\hookrightarrow\ldots\label{eq:12}
\end{equation}
The mathematics of stochastic differential equations does not have
a placeholder for indicating the agent that makes the choices. This
makes the mathematical theory generally applicable to many different
contexts where the agent could be either \emph{internal} or \emph{external}
to the system.

\begin{example}
(Genuinely stochastic dynamics) For the particular context when the
agent is the modeled system itself, we say that the dynamics is \emph{genuinely}
stochastic and the modeled system possesses free will because it is
able to make choices from a set of physically possible future courses
of action.
\end{example}

\begin{example}
(Effectively stochastic dynamics) In the context when the agent is
external to the modeled system, we say that the dynamics is \emph{effectively}
stochastic but the modeled system does not posses free will because
it is unable to make choices. For example, consider two human players
playing a chess game. When only the chess pieces are described, they move stochastically on the
chess board because the two human players make choices and use the
muscles of their hands to move the chess pieces. Since each chess
piece is not an active agent itself, it exhibits effective stochastic
dynamics where the word ``effective'' means that it only ``looks
like'' stochastic dynamics due to the fact that the external
cause is left out from the description. Indeed, if the two human players
are included in the description and their choices are described by stochastic
processes, then the motion of the chess pieces is completely deterministic
as they ``copy'' exactly the outcome of the external human choice. The determinism
of the copying action is due to the presence of conditional probabilities
that are only selected from the set $\{0\%,100\%\}$.
\end{example}

\begin{example}
(Simulated stochastic dynamics) In the special context when we are
the external agent to a studied genuinely stochastic system, we say
that we perform a simulation. In the process of solving a system of
stochastic differential equations (SDEs), we make weighted choices
ourselves and attribute the obtained results to the simulated system
because it \emph{could have produced} genuinely the same outcomes
with the same probability. Of course, when we use computer programs
to perform the simulations we do not even make weighted choices ourselves,
but rather relegate the task to \emph{pseudorandom number generator},
which is perfectly deterministic process that produces outcomes with
the required statistics. Because we do not know how the pseudorandom
number generator works and do not control directly its initialization
state, we pretend that our ``ignorance'' of the outcome of
the deterministic pseudorandom number generation is a good enough
substitution for the genuine stochastic process that generates truly
random numbers.
\end{example}

In the physical world, only quantum systems could exhibit genuinely
stochastic behavior. The Brownian motion of a classical particle in
a fluid could be effectively modeled with a stochastic differential
equation (SDE), but the resulting trajectory is deterministic if the
exact positions and velocities of the rest of the particles in the
fluid are taken into account \cite{Langevin1908}. In other words,
all the information that is needed to predict exactly the trajectory
of the Brownian particle is available somewhere in the environment.
If one wants to transmit securely a secret message along a cryptographic
channel, it would be unwise to rely on effective stochasticity because
someone may find a way to extract the hidden information from the
environment. Instead, one could use a genuine quantum system, which
is guaranteed to produce truly random numbers \cite{Gabriel2010,Zhang2021}.
Thus, quantum systems are the only physical systems that could exhibit
genuinely stochastic dynamics and manifest their free will \cite{Georgiev2021}.

The action of the conscious mind on its own brain is no longer mysterious
in quantum reductionism. In quantum functionalism, it was assumed
that it is the brain that produces the mind and then the mind should
somehow find a way to affect its own brain. In quantum reductionism,
the roles are reversed and in a fundamentally mental physical world
it is the mind that produces the ``observable brain''. The
conscious mind is the set of all physical potentialities that could
be actualized during quantum measurement, whereas the observable brain
is the classical information that characterizes the actual outcome
that has been chosen by the mind. Thus, it is logically impossible
for the actualized mind's choice not to be that mind's choice. For
example, if Alice's observable brain is the actualized choice of Alice's
mind, it is impossible for Alice to choose Bob's observable brain,
because Alice's mind is not Bob's mind.

The most challenging problem met by quantum functionalism (cf. Section~\ref{sub:41}) was the
empirical fact that we do not knowingly choose our observable brain
state. When we choose to move our hand, we experience a desire to
voluntary move our hand without knowing which neuron in our motor
brain cortex is firing electrically. In quantum reductionism, there
is no mystery to such introspective testimony because the quantum
jumps given in Eq.~\eqref{eq:12} describe a stochastic dynamics from
one conscious experience to another conscious experience.
For example, we have a desire to consciously move our hand and then we evolve into a conscious
state in which we have triggered the hand motion. The conscious state
with triggered hand motion does not also have extra self-referential information
of how our brain looks like from a third-person perspective.
Quantum measurement theory \cite{Braginsky1992,Busch1996} further makes it clear that a quantum state does not observe itself because a quantum state cannot measure itself.
The quantum measurement generates communicable information, i.e. classical bits of information \cite{Georgiev2017}.
When the conscious mind experiences certain qualia, e.g. the redness of a red rose,
no communicable information is generated with respect to the phenomenological nature of those qualia, e.g. what is it like to experience the redness of a red rose.
Therefore, the conscious mind \emph{experiences} itself, but \emph{does not observe} itself in the technical sense of the word ``observation'' understood as ``physical measurement'' \cite{Georgiev2020a,Georgiev2020b}.
Furthermore, the conscious mind should not have an extra knowledge of what the
observable brain is, because the brain picture is dependent on the measuring instrument.
For example, microscopic image of pyramidal neurons in the brain cortex
could be taken at different magnification and with different resolution,
the ongoing electrophysiological processes in dendrites, axons and synaptic junctions could be described in molecular
language \cite{Johnston1995,Melkikh2019}, and so on.
Holevo's theorem in quantum information theory puts a strict upper bound on
the amount of classical information that can be obtained by an external
observer from a given quantum state $\vert \Psi\rangle$. Thus, the ``observable
brain'' is nothing but the physically admissible upper limit
of classical information that can be communicated to an external observer
with regard to each and every actualized choice of the dynamically
evolving conscious mind $\vert \Psi(t)\rangle$.

\section{Free will and stochastic processes}
\label{S5}

\begin{definition}
(Free will) \emph{Free will} is the inherent capacity of a physical
system to make genuine choices among at least two physical outcomes
\cite{Georgiev2017,Georgiev2021}.
\end{definition}

The \emph{act of choosing} always selects \emph{one actualized outcome}
from \emph{several available possible outcomes}. Consequently, the
information obtained from a \emph{single choice} is insufficient for
an external observer to differentiate between a deterministic system
and a system endowed with free will. This is because the single choice
produces a single outcome without any accompanying evidence of the
physical existence of the other potential alternative outcomes. Instead,
the external observer needs to collect information from \emph{repeated
choices} performed by the target system and accumulate a statistical
probability distribution for different actualized outcomes. If the
probability distribution is 100\% peaked onto a single outcome, then
the system is guaranteed to be deterministic. However, if the probability
distribution is spread over several outcomes, then this could be a
manifestation of the target system making genuine choices among those
several outcomes. Whether or not the target system actually makes
choices can only be answered by the physical laws that describe the
properties of physical reality. If the physical laws allow/disallow
the capacity for making genuine choices, then real physical systems
can/cannot make genuine choices. The mathematical representation of the
\emph{act of choosing} necessitates the adoption of a \emph{generalized
type of processes} known as \emph{stochastic processes}.

\begin{definition}
(Stochastic process) A \emph{stochastic process}~$\Gamma$ is a generalized
type of process that implements
the actualization of a single outcome $x_{a}$ selected from a sample
space $X=\{x_{1},x_{2},\ldots,x_{n}\}$, where the probability weights
for the alternative outcomes are given by corresponding probability
distribution $P=\{p_{1},p_{2},\ldots,p_{n}\}$ with normalized sum $\sum_{n}p_{n}=1$.
\end{definition}

In modern mathematical software, the function that generates an instance
of a stochastic process is often called \emph{weighted random choice}.
Typically, the latter terminology does not create any confusion because
the majority of mathematical models of stochastic phenomena focus
on the presence of external noise due to uncontrollable factors in
the environment. The attachment of the word \emph{random} to the phrase
\emph{weighted choice}, however, is misleading and may obscure the
fact that a stochastic process~$\Gamma$ can generate a completely
deterministic trajectory. In fact, the most important ingredient in
the definition of the stochastic process~$\Gamma$ is the probability
distribution $P=\{p_{1},p_{2},\ldots,p_{n}\}$ and it can continuously
vary from a completely biased distribution peaked onto a single outcome
to completely unbiased distribution equally spread over all possible
outcomes. Therefore, without having any information about the probability
distribution $P=\{p_{1},p_{2},\ldots,p_{n}\}$, the \emph{stochastic
process}~$\Gamma$ \emph{can be any process} including a completely
deterministic one.

\begin{example}
(Stochastic processes with different bias) The simplest stochastic
process~$\Gamma$ can be realized with a sample space consisting of
only two outcomes $X=\{0,1\}$, each of which can be realized with
corresponding probability weight given by $P=\{p_{0},p_{1}\}$. The
bias in the two-outcome probability distribution can be quantified as
\begin{equation}
\mathscr{B}=p_{1}-p_{0}\label{eq:bias}
\end{equation}
with range $\mathscr{B}\in[-1,1]$, where $\mathscr{B}=1$ indicates
a completely biased (deterministic) distribution such that $x=1$
always occurs, $\mathscr{B}=-1$ indicates a completely biased (deterministic)
distribution such that $x=1$ never occurs, and $\mathscr{B}=0$ indicates
a completely unbiased (indeterministic) distribution such that $x=0$
and $x=1$ are equally likely to occur. Individual stochastic runs
of $n=100$ repetitions of the weighted choice with different values
of the bias $\mathscr{B}$ are shown in Fig. \ref{fig:6}.
Comparative analysis of the obtained stochastic plots for different values of
the bias $\mathscr{B}$ clearly demonstrates that the number of transitions
between $x=0$ and $x=1$ decreases with the increase of the bias towards unity, $\mathscr{B}\to1$,
and in the limit of complete bias, $\mathscr{B}=1$,
the trajectory generated by the stochastic process is completely deterministic.
The Kolmogorov complexity and algorithmic randomness (incompressibility)
of the realized sequences of 0s and 1s also decrease as $\mathscr{B}\to1$
due to the appearance of longer and longer strings of consecutive~1s.
\end{example}

\begin{figure}
\begin{centering}
\includegraphics[width=\textwidth]{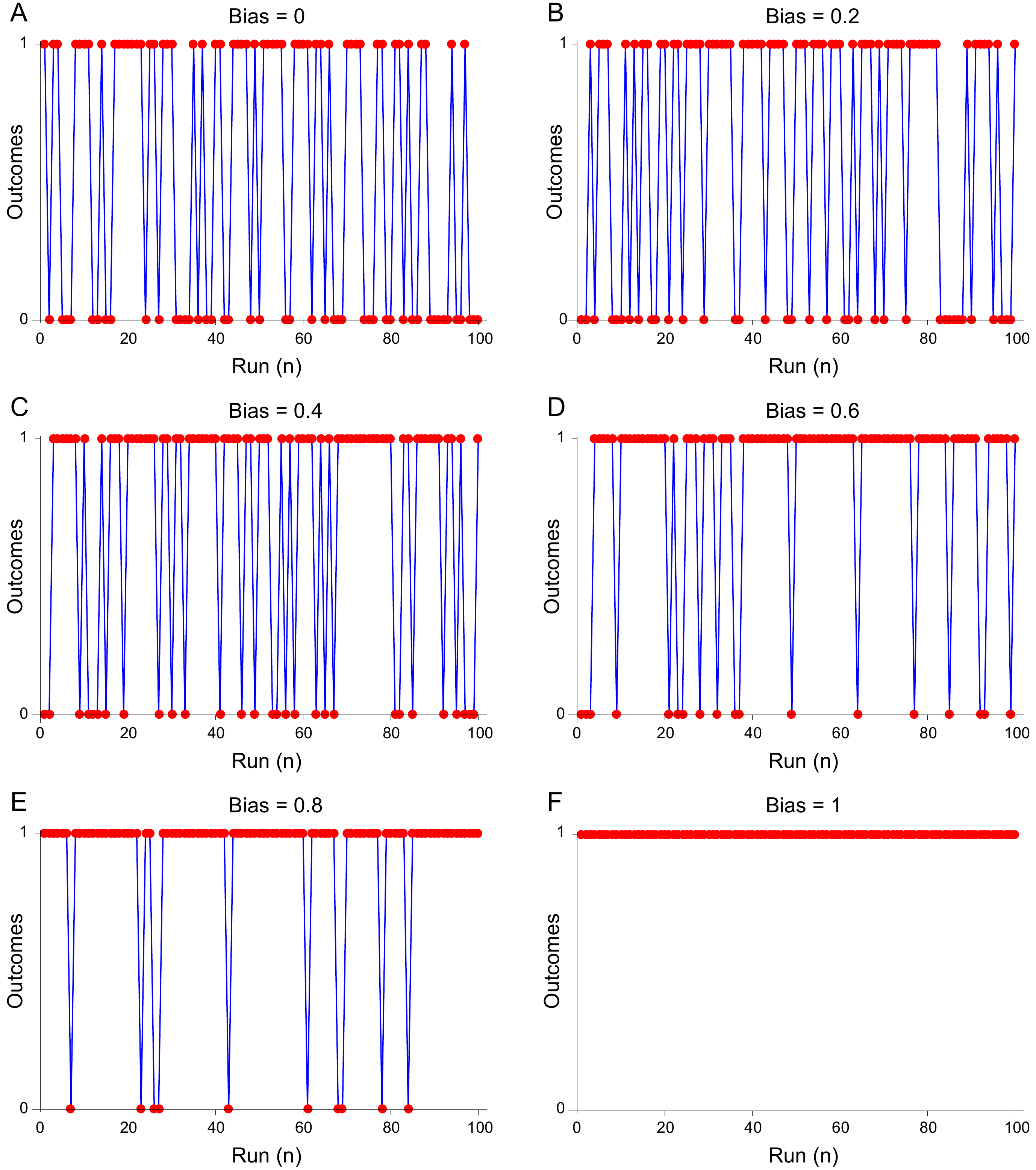}
\par\end{centering}
\caption{\label{fig:6}Two-outcome stochastic processes~$\Gamma$ with different
levels of the bias $\mathscr{B}$ in favor of the outcome $x=1$ over
the outcome $x=0$. Each weighted choice was repeated for $n=100$
times. The individual outcomes at each time point are shown as red
points, whereas consecutive outcomes are connected with a blue line
that serves as a visual guide for transitions between the two outcome
values. For different levels of the bias $\mathscr{B}$ the stochastic
process~$\Gamma$ can produce any trajectory from completely random
to completely deterministic.}
\end{figure}

Now, after we have shown that the term \emph{stochastic}
does not by itself imply anything concrete with regard to \emph{randomness}
or \emph{determinism} without taking into account the actual values
of the \emph{bias} $\mathscr{B}$, we are ready to discuss how the
free will is constrained by the presence of non-zero bias $\left|\mathscr{B}\right|\neq0$.
Also, we will show that randomness is the external manifestation of
free will as perceived by observers outside the target agent.

\begin{example}
(External manifestation of free will) Suppose that Alice and Bob are
agents endowed with free will and each one performs a series of $n=100$ completely
unbiased choices of $X=\{0,1\}$ with $P=\{p_{0}=\frac{1}{2},p_{1}=\frac{1}{2}\}$,
hence $\mathscr{B}=0$. From the perspective of Bob, the sequence
of outcomes chosen by Alice is going to be statistically random due
to approximately equal occurrence of 0s and 1s. From the point of
view of Alice, there is nothing random but only a manifestation of
her own free will because she is the agent who made each of the choices
and for each choice she was able to choose otherwise. When Bob performs
his series of $n=100$ completely unbiased choices, the roles are
reversed, namely, it was Bob who made the choices and manifested his
free will, while Alice perceives Bob's choices to be random. The sequences
of chosen outcomes produced by Alice or Bob are statistically indistinguishable
in terms of \emph{algorithmic randomness} or \emph{incompressibility},
which is maximal. Thus, \emph{randomness} by itself as a statistical
property of the generated outcomes is not something that is incompatible
with \emph{free will}. What is important for the attribution of free
will is the location of the physical source of stochasticity, whether
it is inside Alice or inside Bob.
\end{example}

\begin{definition}
(Amount of free will) The \emph{amount of free will} $\mathscr{F}$
of a target agent can be quantified by an external observer as the
expected information gain in bits from learning the actualized outcome
chosen by the target agent \cite{Georgiev2021}
\begin{equation}
\mathscr{F}=-\sum_{i}p_{i}\log_{2}p_{i}\label{eq:freewill}
\end{equation}
\end{definition}

The amount of free will $\mathscr{F}$ is strongly affected by the
presence of an inherent bias $\mathscr{B}$, which leads to greater
preference by the agent for one outcome over alternative outcomes
\cite{Georgiev2021}. For the simplest stochastic process~$\Gamma$
with a sample space consisting of only two outcomes $X=\{0,1\}$,
the individual probabilities can be expressed in terms of the bias
\eqref{eq:bias} as follows
\begin{align}
p_{0} & =\frac{1-\mathscr{B}}{2}\\
p_{1} & =\frac{1+\mathscr{B}}{2}
\end{align}
For $\mathscr{B}=0$, the amount of free will $\mathscr{F}$ is maximal
\begin{equation}
\mathscr{F}\left(\mathscr{B}=0\right)=-2\times\frac{1}{2}\log_{2}\frac{1}{2}=1
\end{equation}
whereas for $\left|\mathscr{B}\right|=1$, the amount of free will
$\mathscr{F}$ is zero
\begin{equation}
\mathscr{F}\left(\left|\mathscr{B}\right|=1\right)=-0\log_{2}0-1\log_{2}1=0
\end{equation}

\begin{theorem}
\label{thm:nF}The amount of free will for $n$ repeated sequential choices of
a stochastic process~$\Gamma$ with $k$ possible outcomes $X=\{x_{1},x_{2},\ldots,x_{k}\}$
with probability weights $P=\{p_{1},p_{2},\ldots,p_{k}\}$ at each
single time is cumulative and adds up to $n\times\mathscr{F}_{0}$,
where $\mathscr{F}_{0}=-\sum_{i=1}^{k}p_{i}\log_{2}p_{i}$ is the
amount of free will for each individual choice.
\end{theorem}
\begin{proof}
Any stochastic processes with $k$ possible outcomes $X=\{x_{1},x_{2},\ldots,x_{k}\}$
at a single time with probability weights $P=\{p_{1},p_{2},\ldots,p_{k}\}$,
repeated at $n$ times can be viewed as a single choice of the entire
history of $n$ outcomes $p_{i_{1}}p_{i_{2}}\ldots p_{i_{n}}$. Thus,
the amount of free will is
\begin{align}
\mathscr{F} & =-\sum_{i_{1}=1}^{k}\sum_{i_{2}=1}^{k}\cdots\sum_{i_{n}=1}^{k}p_{i_{1}}p_{i_{2}}\ldots p_{i_{n}}\log_{2}\left(p_{i_{1}}p_{i_{2}}\ldots p_{i_{n}}\right)\nonumber \\
 & =-n\left(p_{1}+p_{2}+\ldots+p_{k}\right)^{n-1}\sum_{i=1}^{k}p_{i}\log_{2}p_{i}\nonumber \\
 & =-n\times1^{n-1}\sum_{i=1}^{k}p_{i}\log_{2}p_{i}=n\times\mathscr{F}_{0}
\end{align}
where we have used the fact that the probability weights are normalized
to unity, $\sum_{i=1}^{k}p_{i}=1$, together with the logarithmic
product property converting log of a product into a sum of logs, namely,
$\log(p_{1}p_{2})=\log p_{1}+\log p_{2}$.
\end{proof}
An important consequence of Theorem \ref{thm:nF} is that as long
as the process is not deterministic, even strongly biased stochastic
processes with very small but non-zero amount of free will per individual
choice can produce arbitrary large changes in the resulting dynamics
if the choices are repeated sufficient number of times. In the sample
stochastic processes illustrated in Fig. \ref{fig:6}, the total amount
of free will exercised by the agent is $\mathscr{F}=100$ bits for
$\mathscr{B}=0$, $\mathscr{F}\approx97$ bits for $\mathscr{B}=0.2$,
$\mathscr{F}\approx88$ bits for $\mathscr{B}=0.4$, $\mathscr{F}\approx72$
bits for $\mathscr{B}=0.6$, $\mathscr{F}\approx47$ bits for $\mathscr{B}=0.8$,
and $\mathscr{F}=0$ bits for $\mathscr{B}=1$.

\section{Learning, biases and free will}
\label{S6}

The capacity for making choices is granted by physical laws. Because
the evolutionary processes are governed by the physical laws of the
universe, it is impossible for an organism to evolve a capacity that
miraculously breaks the physical laws. What can be evolutionary achieved,
however, is to acquire strategies for optimal utilization of the physical
laws \cite{Georgiev2017,Georgiev2021}.

In quantum mechanics, the measurement of quantum observables can generate
the whole range of quantum probability distributions varying from
completely biased to completely unbiased, where each probability distribution
grants a different amount of free will $\mathscr{F}$ according to
\eqref{eq:freewill}. From the Born rule, $P=\langle\Psi|\hat{A}|\Psi\rangle$
it is evident that the quantum probabilities change either by changing
the quantum state of the measured system $|\Psi\rangle$ or changing
the basis in which quantum measurement is performed, where the measurement
basis is fixed by the eigenvectors of the measured quantum observable
$\hat{A}$. Because typically the measured observable $\hat{A}$ and
the measurement basis will be fixed by the environment, the organisms
can acquire molecular mechanisms that modify the quantum state $|\Psi\rangle$
of their nervous system that undergoes repeated measurement by the
environment. The preparation of different quantum state $|\Psi(t)\rangle$
for subsequent measurement could be viewed as a form of \emph{knowledge
acquisition} or \emph{learning} as the time $t$ progresses. Next,
we will present the simplest quantum toy example involving a single
qubit that is measured to produce a two-outcome sample space $X=\{0,1\}$.

\begin{figure}[t!]
\begin{centering}
\includegraphics[width=\textwidth]{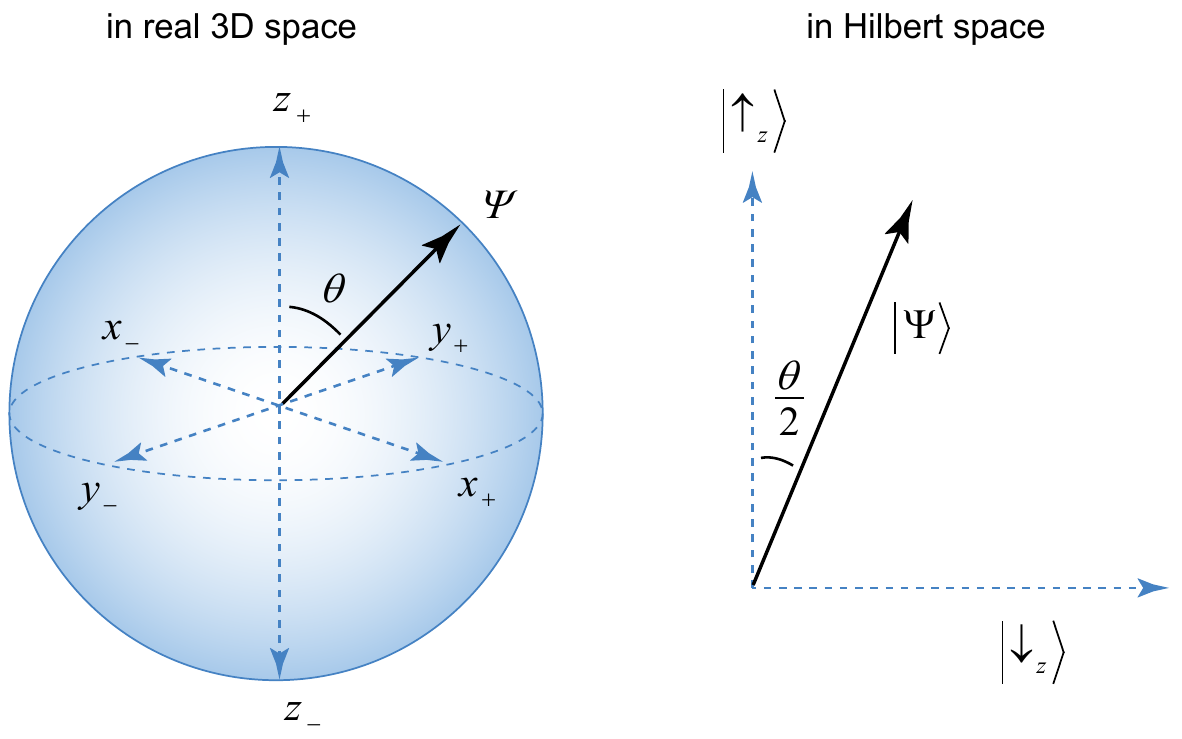}
\par\end{centering}
\caption{\label{fig:7}Quantum toy example illustrating the physical realization
of a stochastic process~$\Gamma$ with two-outcome sample space $X=\{0,1\}$,
eigenbasis of the quantum measurement given by $\{|\downarrow_{z}\rangle,|\uparrow_{z}\rangle\}$
and quantum probability distribution $P=\{p_{0},p_{1}\}$ given by
the Born rule.}
\end{figure}

\begin{example}
(Dynamic biases and free will)
Let the target system be a qubit (spin-$\frac{1}{2}$ particle) whose
quantum state $|\Psi\rangle$ resides in a two-dimensional complex
Hilbert space $\mathcal{H}$. Let the measurement by the environment
determine the orientation of the spin along a fixed axis, which we
can call the $z$-axis, in the real 3-dimensional space. The eigenbasis
of the measurement is given by $\{|\downarrow_{z}\rangle,|\uparrow_{z}\rangle\}$
and the two-outcome sample space $X=\{0,1\}$ is given by the eigenvalues
corresponding to each eigenvector of the measured observable
$\hat{A}= 1 |\uparrow_{z}\rangle\langle\uparrow_{z}| + 0 |\downarrow_{z}\rangle\langle\downarrow_{z}|$.
If the spin of the qubit points away at an angle $\theta$ from the
$z_{+}$-axis in the real 3D space (Fig.~\ref{fig:7}), the quantum
state of the qubit in Hilbert space can be expressed as
\begin{equation}
|\Psi\rangle=\cos\left(\frac{\theta}{2}\right)|\uparrow_{z}\rangle+\sin\left(\frac{\theta}{2}\right)|\downarrow_{z}\rangle
\end{equation}
where we have used the freedom to choose which direction in 3-dimensional space will be called $x_+$ thereby removing an inessential pure phase factor from one of the superposed states. The Born rule determines the two probability weights for the actualization
of each of the two possible outcomes 
\begin{align}
p_{0} & =\langle\Psi|\downarrow_{z}\rangle\langle\downarrow_{z}|\Psi\rangle=\sin^{2}\left(\frac{\theta}{2}\right)\\
p_{1} & =\langle\Psi|\uparrow_{z}\rangle\langle\uparrow_{z}|\Psi\rangle=\cos^{2}\left(\frac{\theta}{2}\right)
\end{align}
The weighted choice performed by the target qubit upon measurement
is described by a stochastic process~$\Gamma$ with two-outcome sample
space $X=\{0,1\}$ and probability weights $P=\{p_{0},p_{1}\}$ given
by the Born rule. If the target qubit is not supported with some form
of memory and it is \emph{repeatedly prepared} in the same state $|\Psi\rangle$
and then \emph{measured} \cite{Gudder2022,Busch1990} in the $\{|\downarrow_{z}\rangle,|\uparrow_{z}\rangle\}$
basis, the stochastic process will exhibit constant bias $\mathscr{B}$
in time
\begin{equation}
\mathscr{B}=\cos^{2}\left(\frac{\theta}{2}\right)-\sin^{2}\left(\frac{\theta}{2}\right)=\cos\theta
\end{equation}
with constant free will in time
\begin{equation}
\mathscr{F}=-\frac{1-\mathscr{B}}{2}\log_{2}\left(\frac{1-\mathscr{B}}{2}\right)-\frac{1+\mathscr{B}}{2}\log_{2}\left(\frac{1+\mathscr{B}}{2}\right)
\end{equation}
Inside the living brain endowed with memory and dopamine reward system,
however, the qubit may participate in knowledge acquisition or learning
in time $t$ and exhibit a dynamic bias $\mathscr{B}(t)$ and dynamic
amount of free will $\mathscr{F}(t)$. For example, if the outcome
$|\uparrow_{z}\rangle$ is followed by reward, the qubit state $|\Psi(\theta,t)\rangle$
can be prepared with angle $\theta(t)\to0$, increasing the bias $\mathscr{B}(t)\to1$
and limiting the amount of free will $\mathscr{F}(t)\to0$ with the
expectation that future measurement of the spin along the $z$-axis
will have an increased probability of choosing the outcome $|\uparrow_{z}\rangle$,
hence again receiving reward. In contrast, if the outcome $|\uparrow_{z}\rangle$
is followed by punishment, the qubit state $|\Psi(\theta,t)\rangle$
can be prepared with angle $\theta(t)\to\pi$, increasing the bias
in the opposite direction $\mathscr{B}(t)\to-1$ and limiting the
amount of free will $\mathscr{F}(t)\to0$ with the expectation that
future measurement of the spin along the $z$-axis will have an increased
probability of choosing the outcome $|\downarrow_{z}\rangle$, hence
avoiding receiving another punishment. In the absence of previous
experience or in the presence of conflicting information about the
consequences of the two possible outcomes, the qubit state $|\Psi(\theta,t)\rangle$
can be prepared with angle $\theta(t)\to\frac{\pi}{2}$, decreasing
the bias $\mathscr{B}(t)\to0$ and increasing the amount of free will
$\mathscr{F}(t)\to1$ so that either outcome can be chosen and the
consequences of each choice can be investigated through accumulation
of new knowledge. In conclusion, the greater certainty in accumulated
knowledge is associated with greater bias $\mathscr{B}(t)$ and lower
amount of manifested free will $\mathscr{F}(t)$.
\end{example}

The single qubit in isolation is necessarily too simple and does not
have the mechanisms of memory or reward. In the living brain, however,
the elementary particles assemble into biomolecules that can store
memories for prolonged periods of time and there are biochemical cascades
that can be triggered by rewards or punishments. For example, voltage-gated
ion channels can be subject to phosphorylation or dephosphorylation
triggered by dopamine release, which modifies the sensitivity of their
voltage-sensors to the transmembrane electric field and changes the
probabilities for the ion channels to be in open or closed state (Fig.~\ref{fig:5}).
The construction of precise quantum models of neuronal function is
beyond the scope of this study and will be the subject of future research.

\section{Quantum entanglement, mind binding and free will}
\label{S7}

The composite quantum state $|\Psi\rangle$ of $k$ components resides
in a tensor product Hilbert space $\mathcal{H}=\mathcal{H}_{1}\otimes\mathcal{H}_{2}\otimes\ldots\otimes\mathcal{H}_{k}$
formed by the individual Hilbert spaces $\mathcal{H}_{1},\mathcal{H}_{2},\ldots,\mathcal{H}_{k}$
of the corresponding components. Inside the composite Hilbert space
$\mathcal{H}$, there are two kinds of states: those that are factorizable
and those that are non-factorizable.

\begin{definition}
(Non-entangled quantum states) Factorizable composite quantum states in the form
\begin{equation}
|\Psi\rangle=|\psi_{1}\rangle\otimes|\psi_{2}\rangle\otimes\ldots\otimes|\psi_{k}\rangle\label{eq:separable}
\end{equation}
are called separable or non-entangled states \cite{Gudder2020a,Gudder2020b}.
\end{definition}

\begin{definition}
(Entangled quantum states) Non-factorizable composite quantum states,
which cannot be expressed as a tensor product of component states
\begin{equation}
|\Psi\rangle\neq|\psi_{1}\rangle\otimes|\psi_{2}\rangle\otimes\ldots\otimes|\psi_{k}\rangle
\end{equation}
are called non-separable or entangled states \cite{Georgiev2022b,Georgiev2022c}.
\end{definition}

In the quantum reductive approach, minds are attributed only to non-factorizable
pure quantum states \cite{Georgiev2017}. This means that quantum
entanglement binds conscious experiences into a single unitary mind.
The components of a quantum entangled state do not have definite individual
state vectors, therefore it is not surprising that they are not endowed
with individual minds. On the other hand, separable quantum states
of the form \eqref{eq:separable} are such that both the composite
system and the component systems have definite state vectors, namely,
the composite state vector is $|\Psi\rangle$, whereas the individual
component state vectors are $|\psi_{1}\rangle,|\psi_{2}\rangle,\ldots,|\psi_{k}\rangle$.
In the latter case, the existence of a state vector such as $|\Psi\rangle$
is not sufficient to guarantee the attribution of a single mind, instead
the separability of $|\Psi\rangle$ makes it a \emph{collection of
minds}. This is more easily understandable by taking into consideration
that separability actually implies \emph{statistical independence} \cite{Georgiev2021b}.
In other words, two separable minds have no direct access to each other's conscious experiences
and possess their own individual free will. The situation can be clarified
with the following simplified toy example.

\begin{example}
(Separate minds have independent free will) Suppose that Alice and
Bob are two separate minds, each of which is modeled by a definite
two-level state
\begin{align}
|\Psi_{A}\rangle & =\alpha_{0}|\downarrow_{z}\rangle+\alpha_{1}|\uparrow_{z}\rangle\\
|\Psi_{B}\rangle & =\beta_{0}|\downarrow_{z}\rangle+\beta_{1}|\uparrow_{z}\rangle
\end{align}
The composite state is 
\begin{equation}
|\Psi_{AB}\rangle=|\Psi_{A}\rangle\otimes|\Psi_{B}\rangle
\end{equation}
which makes the resulting quantum probabilities $P=\{p_{00},p_{01},p_{10},p_{11}\}$
for each of the four possible outcomes $\{|\downarrow_{z}\downarrow_{z}\rangle,|\downarrow_{z}\uparrow_{z}\rangle,|\uparrow_{z}\downarrow_{z}\rangle,|\uparrow_{z}\uparrow_{z}\rangle\}$
to be factorizable 
\begin{equation}
p_{00}=\left|\alpha_{0}\right|^{2}\left|\beta_{0}\right|^{2},\quad p_{01}=\left|\alpha_{0}\right|^{2}\left|\beta_{1}\right|^{2},\quad p_{10}=\left|\alpha_{1}\right|^{2}\left|\beta_{0}\right|^{2},\quad p_{11}=\left|\alpha_{1}\right|^{2}\left|\beta_{1}\right|^{2}
\end{equation}
This means that the amount of free will $\mathscr{F}_{AB}$ manifested
by the composite system comprised of Alice and Bob is exactly the
sum of the amount of free will $\mathscr{F}_{A}$ manifested by Alice
and the amount of free will $\mathscr{F}_{B}$ manifested by Bob 
\begin{align}
\mathscr{F}_{AB} & =-\sum_{i=0}^{1}\sum_{j=0}^{1}\left|\alpha_{i}\right|^{2}\left|\beta_{j}\right|^{2}\log_{2}\left(\left|\alpha_{i}\right|^{2}\left|\beta_{j}\right|^{2}\right)\nonumber \\
 & =-\left(\left|\beta_{0}\right|^{2}+\left|\beta_{1}\right|^{2}\right)\sum_{i=0}^{1}\left|\alpha_{i}\right|^{2}\log_{2}\left|\alpha_{i}\right|^{2}-\left(\left|\alpha_{0}\right|^{2}+\left|\alpha_{1}\right|^{2}\right)\sum_{j=0}^{1}\left|\beta_{j}\right|^{2}\log_{2}\left|\beta_{j}\right|^{2}\nonumber \\
 & =-\sum_{i=0}^{1}\left|\alpha_{i}\right|^{2}\log_{2}\left|\alpha_{i}\right|^{2}-\sum_{j=0}^{1}\left|\beta_{j}\right|^{2}\log_{2}\left|\beta_{j}\right|^{2}=\mathscr{F}_{A}+\mathscr{F}_{B}
\end{align}
where we have used the logarithmic product property and the normalization
of the component states, namely, $\left|\alpha_{0}\right|^{2}+\left|\alpha_{1}\right|^{2}=1$
and $\left|\beta_{0}\right|^{2}+\left|\beta_{1}\right|^{2}=1$.
\end{example}

\begin{example}
(Entanglement binds components into a single mind) Consider now what
would happen if the composite state of two components $A$ and $B$
was maximally quantum entangled
\begin{equation}
|\Psi_{AB}\rangle=\frac{1}{\sqrt{2}}\left(|\downarrow_{z}\uparrow_{z}\rangle+|\uparrow_{z}\downarrow_{z}\rangle\right)
\end{equation}
In the quantum entangled state, neither component $A$ nor component $B$
can have their own state vector and the resulting quantum probabilities
$P=\{p_{00},p_{01},p_{10},p_{11}\}$ for each of the four possible
outcomes $\{|\downarrow_{z}\downarrow_{z}\rangle,|\downarrow_{z}\uparrow_{z}\rangle,|\uparrow_{z}\downarrow_{z}\rangle,|\uparrow_{z}\uparrow_{z}\rangle\}$
are not factorizable
\begin{equation}
p_{00}=0,\quad p_{01}=\frac{1}{2},\quad p_{10}=\frac{1}{2},\quad p_{11}=0
\end{equation}
The reduced density matrices for $A$ or $B$ are completely mixed
\begin{equation}
\hat{\rho}_{A}=\hat{\rho}_{B}=\left(\begin{array}{cc}
\frac{1}{2} & 0\\
0 & \frac{1}{2}
\end{array}\right)
\end{equation}
which means that when viewed locally the probability weights for both
$A$ and $B$ are $P=\{\frac{1}{2},\frac{1}{2}\}$ for the outcomes
$\{|\downarrow_{z}\rangle,|\uparrow_{z}\rangle\}$.
Without taking into consideration the quantum entanglement, from the local
statistics one may incorrectly conclude that component~$A$ manifested $\mathscr{F}_{A}=1$
bit of free will and component~$B$ also manifested $\mathscr{F}_{B}=1$
bit of free will. When the quantum entanglement is taken into account,
however, it can be correctly concluded that the composite system exhibited only
$\mathscr{F}_{AB}=1$ bit of free will because the outcomes produced
by $A$ and $B$ were perfectly anti-correlated, when $A$ chose $|\downarrow_{z}\rangle$
it was always the case that $B$ chose $|\uparrow_{z}\rangle$ and
when $A$ chose $|\uparrow_{z}\rangle$ it was always the case that
$B$ chose $|\downarrow_{z}\rangle$. This apparent \emph{subadditivity}
of the manifested amount of free will, namely, $\mathscr{F}_{AB}=1<1+1=2=\mathscr{F}_{A}+\mathscr{F}_{B}$
is due to the \emph{quantum correlations} resulting from quantum entanglement.
In other words, the component subsystems cannot manifest their own
free will independently from the rest of the composite quantum entangled
state. Instead, it is only the composite quantum entangled system
as a whole that manifests its own independent free will and imposes
quantum correlations on the components. The attribution of free will
to the composite entangled system, which has a non-factorizable state
vector and a mind, but not to the component subsystems which have
neither their own state vectors nor independent minds, provides a
one-to-one correspondence between non-factorizable pure quantum state
vectors, single minds and their independent free will \cite{Georgiev2017}. To summarize,
pure quantum entangled systems have a single mind and it is the composite
quantum entangled system as a whole that possesses the free will.
\end{example}

\section{Wave function collapse and disentanglement}
\label{S8}

The capacity of quantum systems to make genuine choices requires a
physical \emph{actualization} process that converts one of the \emph{possible}
outcomes into an \emph{actual} outcome. It is exactly the act of actualization
that makes the stochastic process~$\Gamma$ suitable for describing
the concept of choosing. Without actualizations there can be only
possibilities without anything ever happening in reality. In quantum
physics, this means that the unitary quantum dynamics prescribed by
the Schr\"{o}dinger equation can only produce possibilities, but it needs
a stochastic process governed by the Born rule in order to create
actualities. The stochastic actualization process is referred to as
\emph{wave function collapse} and is expected to occur when the composite
quantum system reaches a certain energy threshold $\mathscr{E}$.
The energy threshold $\mathscr{E}$ for wave function collapse is
a parameter to be determined experimentally \cite{Bassi2003,Vinante2020,Carlesso2022}, but in the context of
quantum reductive theories of consciousness it is expected to be way
above the energies of individual elementary particles and slightly
below the total energy consumed by the metabolically active human
brain \cite{Georgiev2017}. The existence of an energy threshold $\mathscr{E}$
for wave function collapse solves the measurement problem in quantum
mechanics and prevents the whole universe into getting entangled into
a single universal cosmic mind \cite{Georgiev2017,Georgiev2020b}.

\begin{example}
(Schr\"{o}dinger's cat) The measurement problem results from the unitary
quantum dynamics prescribed by the Schr\"{o}dinger equation
\begin{equation}
\imath\hbar\frac{d}{dt}|\Psi(t)\rangle=\hat{H}|\Psi(t)\rangle
\end{equation}
where $\hat{H}$ is the Hamiltonian and $|\Psi(t)\rangle$ is the
time-dependent quantum state vector of the system. The Schr\"{o}dinger equation is a linear
ordinary differential equation (ODE) with formal solution
\begin{equation}
|\Psi(t)\rangle=e^{-\frac{\imath}{\hbar}\hat{H}t}|\Psi(0)\rangle=\hat{U}|\Psi(0)\rangle
\end{equation}
The linearity implies that if $|\Psi_{1}\rangle$ and $|\Psi_{2}\rangle$
are two solutions of the Schr\"{o}dinger equation, then any linear combination
of those two solutions is also a valid solution
\begin{equation}
|\Psi\rangle=\alpha_{1}|\Psi_{1}\rangle+\alpha_{2}|\Psi_{2}\rangle
\end{equation}
The unitary time evolution operator $\hat{U}=e^{-\frac{\imath}{\hbar}\hat{H}t}$,
which governs the quantum dynamics, is also linear 
\begin{equation}
\hat{U}\left(\alpha_{1}|\Psi_{1}\rangle+\alpha_{2}|\Psi_{2}\rangle\right)=\alpha_{1}\hat{U}|\Psi_{1}\rangle+\alpha_{2}\hat{U}|\Psi_{2}\rangle
\end{equation}
In order to see how the unitary quantum dynamics leads to quantum
entangled superpositions of macroscopic devices, consider a single
photon $|\gamma\rangle$, which if prepared in a state with horizontal
polarization $|\gamma_{H}\rangle$ is detected by a macroscopic measuring
device initially prepared in state $|M\rangle$ whose pointer moves horizontally to final state $|M_{H}\rangle$
and if prepared in a state with vertical polarization $|\gamma_{V}\rangle$
is detected by the same macroscopic measuring device $|M\rangle$
but the pointer moves vertically to final state $|M_{V}\rangle$. For each of the
two alternative cases, the action of the unitary operator governing
the quantum dynamics of the composite system composed of the photon
and the measuring device can be written as
\begin{align}
\hat{U}|\gamma_{H}\rangle|M\rangle & =|\gamma_{H}\rangle|M_{H}\rangle\\
\hat{U}|\gamma_{V}\rangle|M\rangle & =|\gamma_{V}\rangle|M_{V}\rangle
\end{align}
Now, consider what will happen if the photon is polarized diagonally
$|\gamma_{+}\rangle=\frac{1}{\sqrt{2}}(|\gamma_{H}\rangle+|\gamma_{V}\rangle)$
before it is sent to the measuring device. The linearity of the unitary
time evolution operator creates an entangled quantum state
\begin{align}
\hat{U}|\gamma_{+}\rangle|M\rangle & =\hat{U}\,\frac{1}{\sqrt{2}}(|\gamma_{H}\rangle+|\gamma_{V}\rangle)\,|M\rangle\nonumber \\
 & =\frac{1}{\sqrt{2}}(\hat{U}|\gamma_{H}\rangle|M\rangle+\hat{U}|\gamma_{V}\rangle|M\rangle)\nonumber \\
 & =\frac{1}{\sqrt{2}}(|\gamma_{H}\rangle|M_{H}\rangle+|\gamma_{V}\rangle|M_{V}\rangle)\label{eq:cat}
\end{align}
Experimentally, we always observe only one of the two possible alternative
readings of the measurement device, but never the entangled superposition
of both. To highlight the conflict between experimental observations
and entangled macroscopic superpositions, Erwin Schr\"{o}dinger proposed
that the measuring device be replaced by a cat placed inside a box
equipped with a photosensitive mechanism that releases poisonous gas
only if the photon is detected in one of the states, e.g., $|\gamma_{H}\rangle$.
In this case, the resulting entangled macroscopic quantum superposition
will include a cat that is simultaneously dead and alive \cite{Schrodinger1935,Schrodinger1983}.

The conflict between quantum mechanics and the lack of observable
macroscopic superpositions is solved by the energy threshold $\mathscr{E}$
for wave function collapse. If the measuring device is sufficiently
large to pass the energy threshold $\mathscr{E}$, the quantum entangled
state \eqref{eq:cat} undergoes stochastic disentanglement~$\Gamma$
to one of the two separable outcomes $\{|\gamma_{H}\rangle|M_{H}\rangle,|\gamma_{V}\rangle|M_{V}\rangle\}$
with probability weights given by the Born rule $P=\{\frac{1}{2},\frac{1}{2}\}$.
Here, we recall that the concise product notation of two kets actually
implies the tensor product, namely, $|\gamma_{H}\rangle|M_{H}\rangle\equiv|\gamma_{H}\rangle\otimes|M_{H}\rangle$
and $|\gamma_{V}\rangle|M_{V}\rangle\equiv|\gamma_{V}\rangle\otimes|M_{V}\rangle$.
\end{example}

Using all of the technical concepts defined above, we are now able
to briefly outline how the quantum reductive model of consciousness
is supposed to operate.

\begin{example}
(Quantum reductive model of consciousness) The model of how our consciousness
inputs sensory information from the environment, makes choices and
outputs those choices to the environment implements a repeated cycle
consisting of several steps. 

\emph{Step 1}. The starting point of the cycle can be defined to be
a part of the anatomical brain cortex that is in a disentangled tensor
product state of component biomolecules, ions or elementary particles,
$|\Psi(t_{1})\rangle=|\psi_{1}(t_{1})\rangle\otimes|\psi_{2}(t_{1})\rangle\otimes\ldots\otimes|\psi_{k}(t_{1})\rangle$.
At this stage, the composite system $|\Psi(t_{1})\rangle$ is a \emph{collection
of elementary minds}. Among the components $\{|\psi_{1}(t_{1})\rangle,|\psi_{2}(t_{1})\rangle,\ldots,|\psi_{k}(t_{1})\rangle\}$
are included energy quanta coming from the environment that input
sensory information from the surrounding world.

\emph{Step 2}. Then, the components interact with each other and undergo
unitary quantum dynamics prescribed by the Schr\"{o}dinger equation. Nearby
components entangle with each other to form small entangled clusters,
then as time goes on these small entangled clusters entangle with
each other to form larger entangled clusters, and so on, until
a sufficiently large entangled cluster $|\Psi(t_{2})\rangle$ reaches
the energy threshold $\mathscr{E}$ for wave function collapse. At
this stage, the quantum entanglement has bound the conscious experiences
of all the component elementary minds into a single integrated conscious
experience $|\Psi(t_{2})\rangle$ that could be recognized as the
\emph{conscious self}. Furthermore, because the energy threshold $\mathscr{E}$
for wave function collapse is reached the \emph{conscious self }has
to make a choice and selects one of the available disentangled outcomes
using its free will.

\emph{Step 3}. After the choice is made, the chosen outcome by the
conscious self is a disentangled state $|\Psi(t_{3})\rangle=|\psi_{1}(t_{3})\rangle\otimes|\psi_{2}(t_{3})\rangle\otimes\ldots\otimes|\psi_{k}(t_{3})\rangle$
such that some small portion of the disentangled components leaves
the brain in order to output motor information towards the environment
and gets replaced by incoming sensory energy quanta, for example,
we can say that $|\psi_{k}(t_{3})\rangle$ was an electric signal
that left the brain and it was replaced by another sensory signal
$|\psi_{k^{\prime}}(t_{3})\rangle$.

\emph{Step 4}. The cycle can be considered complete when the growth
of entangled clusters proceeds based on unitary quantum dynamics
and another sufficiently large entangled cluster $|\Psi(t_{4})\rangle$
reaches the energy threshold $\mathscr{E}$ for wave function collapse.
This future version of the conscious self is composed by the same
component subsystems indexed from $1$ to $k-1$ together with the
new sensory component $k^{\prime}$ and lacking the motor output component
$k$ sent to the environment. Thus, the single integrated conscious
experience $|\Psi(t_{4})\rangle$ is another instance of the conscious
self that is aware of the new sensory information and bears the consequences
of its past motor choice sent to the environment. 

This quantum reductive cycle of consciousness has been constructed
in such a way, so that it addresses several philosophical problems of
consciousness at once (the mind--brain problem, the mind physical
boundary problem, the mind binding problem, the mind causal potency
problem, the free will problem, the mind inner privacy problem and
the hard problem) and explains how the conscious self changes in time
by making choices and how it can acquire knowledge based on the outcomes
of past choices (for a detailed exposition and more extensive discussion,
see Chapter~6 in Ref.~\citen{Georgiev2017}).
\end{example}

\section{Discussion}
\label{sec:9}

Our prehistoric ancestors have painted cave walls with exquisite art
images of horses, bisons, mammoths, cave bears, lions, panthers, rhinoceroses,
owls and other animals \cite{Chauvet1996,Chauvet2001,Aubert2014,Hoffmann2018,Pike2012,Herzog2010,Harris2011}. The anatomical
description of these animals, some of which are long extinct, is precisely
captured by the prehistoric artists \cite{Kurten1972,Bon2011}. Because
the primary purpose of art is to elicit conscious experiences in the
viewer, it is highly implausible to assume that the prehistoric artists
have produced those paintings without having any experiences of the
animals that were painted. Furthermore, due to the precise anatomical
correspondence between the animals that were seen and the animals that were
painted, it is fair to conclude that the conscious experiences of
the prehistoric painters have been causally potent in producing
the physical artifacts in the form of cave paintings \cite{Lewis2002,Lewis2003}.
The causal potency of consciousness in the physical world is a prerequisite
for the evolution of human consciousness through natural selection
of our animal ancestors. Yet, if one sticks to the premises of classical
functionalism, it would appear that the conscious mind is causally
impotent in the physical world and could not have evolved naturally.
Therefore, one needs to reject classical functionalism and search for better physical alternatives.

Here, we have presented a thorough analysis of the ramifications of two contrasting approaches to consciousness, namely, functionalism or reductionism, within the two main physical frameworks provided by classical physics or quantum physics.
To derive rigorous theoretical results about classical or quantum systems, we have comprehensively reviewed the mathematical theory behind ordinary differential equations (ODEs) or stochastic differential equations (SDEs) and have illustrated the characteristic properties of the resulting classical or quantum dynamics using minimal toy systems.
Further consideration of the mind--brain relationship within each of the four theoretical schemes, namely,
classical functionalism, classical reductionism, quantum functionalism
or quantum reductionism, revealed that only the latter scheme supports
both the causal potency of the conscious mind and free will.

Quantum reductionism endorses a form of panexperientialism, according to which elementary feelings are attributed to all quantum systems, including simple living organisms that do not have a nervous system. Such elementary experiences or elementary feelings would be also attributed to inanimate quantum materials with the explicit caveat that these are memoryless. In fact, our adult human consciousness is a combination of ongoing sensory experiences and a recalled memory of who we are. In medical practice, we have the striking observations of a newborn baby or an aged person with Alzheimer's disease, who definitely have conscious experiences, but either do not yet know who they are or have forgotten who they are. Without the ability to memorize past experiences, a quantum physical system will not be able to communicate to the surrounding world that it had those experiences. Therefore, the presence of transient memoryless experiences in inanimate quantum materials is not something to be bothered by, instead these are the primordial substrates for natural evolution from which the conscious experiences in living systems should have evolved, given the fact that living organisms use free energy and are capable of storing memories. During sleep, the energy consumption by the brain is rapidly suppressed, which is reflected by the fact that the conscious experiences during dreaming are often illogical and hard to recall. This further highlights the need of energy source for storing memories and recalling memories of past experiences, as opposed to merely having any form of experience.

Alternative quantum functional models of consciousness rely on the problematic idea that conscious experiences are somehow generated by an insentient brain substrate. For example, the Orch OR model proposed by Penrose and Hameroff postulates that the quantum brain dynamics is unconscious or subconscious, whereas it is the event of objective reduction that creates intermittently the ``flashes'' of conscious experiences with a frequency of~40~Hz. The problem with the latter claim is the reversed causal order, namely, if the event of objective reduction is causally potent in creating the phenomenological content of the conscious experience, then the generated conscious experience would be causally impotent to choose the outcome of the objective reduction. In contrast, the quantum reductionism advocated in this work holds that the quantum dynamics is continuously conscious and it is the ongoing consciousness that is causally potent in choosing the actual outcome of the objective reduction.

Philosophers often present together the problems of mental causation and free will \cite{Sartorio2016}.
In fact, the two problems are distinct but hierarchically organized as follows:
(1) causally impotent mind cannot cause any course of action in the physical world,
(2) causally potent mind that lacks free will can cause exactly one course of action in the physical world, and
(3) causally potent mind that possesses free will can choose between two or more possible courses of action in the physical world.
This hierarchy of the two problems implies that one could solve the problem of mental causation
without solving the free will problem, whereas the converse is impossible,
namely, solving the free will problem cannot be done without also
solving the problem of mental causation.

In ethics, the hierarchical organization of the problems of mental
causation and free will leads to different levels of blame attribution.
If your conscious mind is causally potent, then you can be
blamed for what you have done. If in addition to causal potency your
conscious mind possesses free will through which to execute choices,
then you can also be blamed for what you have not done. The logic of blame
attribution is straightforward: (1) if your conscious mind is causally
potent but lacks free will, then there is only a single course of
action that you are able to cause in the physical world. Since you
could not have done otherwise, you can only be blamed for what you
have done, but not for what you have not done. Alternatively, (2)
if your conscious mind possesses free will, then you are able to make
a genuine choice among at least two possible courses of action. This
immediately makes you morally responsible not only for what you have
chosen to do, but also for what you have chosen not to do. In other
words, the existence of free will comes with the burden of having
to contemplate the consequences of your actions so that you not only
avoid causing harm, but also avoid missing potential benefits had
you acted differently.

In evolution theory, the lack of causal potency is much more harmful
proposition compared to the lack of free will because a causally impotent
human consciousness could not have evolved through natural selection.
Consequently, the main goal of the present study was to derive rigorously
the causal potency or impotency of consciousness within the four theoretical
schemes including functionalism or reductionism in classical or quantum physics.
Noteworthy, utilizing the mathematical properties of ordinary differential equations,
we have demonstrated that classical functionalism is not a plausible
physical theory of consciousness and classical physics is a disadvantageous
framework for approaching the mind--brain problem. Then, we have established
that quantum physics is not just an exotic extravagance, but an indispensable general
theoretical framework supporting stochastic differential equations
that are ideally tailored to describe dynamic trajectories produced
by sequential choices. Further building upon the quantum information
theoretic properties of quantum states and quantum observables, we
have shown that quantum reductionism predicts unobservable conscious
mind that is causally potent in choosing the future course of action
of the observable brain. This explains how the human consciousness
could have evolved through natural selection in our animal ancestors.

\interlinepenalty=10000

\end{document}